%% file: main.tex
\documentclass[11pt]{article} 
\usepackage{fullpage}

\usepackage[utf8]{inputenc} 
\usepackage[T1]{fontenc}    
\usepackage{url}            
\usepackage{booktabs}       
\usepackage{amsfonts}       
\usepackage{nicefrac}       
\usepackage{microtype}      
\usepackage{xcolor}         

\usepackage{times}
\usepackage{tikz}
\usepackage{mathtools}
\usepackage{amsthm}
\usepackage{amssymb}
\usepackage{bbm}
\usepackage{verbatim}
\usepackage{enumitem}
\usepackage{xspace}
\usepackage{todonotes}
\usepackage{amsfonts,mathrsfs,color}
\usepackage{xcolor}
\usepackage{graphicx}
\usepackage{booktabs}
\usepackage{nicefrac}
\usepackage{mdframed}
\usepackage{contour}
\contourlength{0.1pt}
\contournumber{10}
\usepackage[
  backend=biber,
  style=alphabetic,
  maxbibnames=10,
  maxalphanames=3,
  minalphanames=3,
  natbib=true,
  url=false
]{biblatex}
\usepackage{multirow}
\usepackage[ruled,linesnumbered,noend]{algorithm2e} 
\SetKwComment{Comment}{/* }{ */}
\SetKwProg{Fn}{function}
\SetAlFnt{\relax}
\SetAlCapFnt{\relax}
\SetAlCapNameFnt{\relax}
\SetAlCapHSkip{0pt}
\usepackage{tikz}
\usepackage{xcolor}

\usepackage{color-edits}
\addauthor[HL]{hl}{magenta}

\newcommand{\declarecolor}[2]{\definecolor{#1}{RGB}{#2}\expandafter\newcommand\csname #1\endcsname[1]{\textcolor{#1}{##1}}}
\declarecolor{White}{255, 255, 255}
\declarecolor{Black}{0, 0, 0}
\declarecolor{Maroon}{128, 0, 0}
\declarecolor{Coral}{255, 127, 80}
\declarecolor{Red}{182, 21, 21}
\declarecolor{LimeGreen}{50, 205, 50}
\declarecolor{DarkGreen}{0, 80, 0}
\declarecolor{Purple}{146, 42, 158}
\declarecolor{Navy}{0, 0, 128}
\declarecolor{LightBlue}{84, 101, 202}
\usepackage{hyperref}
\hypersetup{
    colorlinks,
    citecolor=DarkGreen,
    linkcolor=LightBlue}
\usepackage[noabbrev, capitalise, nameinlink]{cleveref}
\input{macro}

\newcommand{\bR}{\overline{\-R}}

\newcommand{\ext}{\mathrm{ext}}
\newcommand{\convex}{\mathrm{c}}
\newcommand{\lsc}{\mathrm{lsc}}
\newcommand{\strc}{\mathrm{sc}}

\newcommand{\Phiprox}{\Phi^{\X}_{\prox}}
\newcommand{\Phiproxeq}{\Phi_{\prox}}

\newcommand{\Dx}{D_{\X}}

\title{Proximal Regret and Proximal Correlated Equilibria: A New Tractable Solution Concept for Online Learning and Games\thanks{Authors are alphabetically ordered.}}

\author{
Yang Cai\thanks{Yale University. Email: \texttt{yang.cai@yale.edu}}\\
\and 
Constantinos Daskalakis\thanks{MIT CSAIL. Email:\texttt{costis@csail.mit.edu}}\\
\and 
Haipeng Luo\thanks{University of Southern California. Email: \texttt{haipengl@usc.edu}} \\
\and
Chen-Yu Wei\thanks{University of Virginia. Email: \texttt{chenyu.wei@virginia.edu}}\\
\and 
Weiqiang Zheng\thanks{Yale University. Email: \texttt{weiqiang.zheng@yale.edu}}\\
}

\date{}
\begin{document}
\maketitle
\begin{abstract}
Learning and computation of equilibria are central problems in game theory, theory of computation, and artificial intelligence.
In this work, we introduce \emph{proximal regret}, a new notion of regret based on proximal operators that lies strictly between external and swap regret. 
When every player employs a no-proximal-regret algorithm in a general convex game, the empirical distribution of play converges to \emph{proximal correlated equilibria} (PCE), a refinement of coarse correlated equilibria. 
Our framework unifies several emerging notions in online learning and game theory---such as gradient equilibrium and semicoarse correlated equilibrium---and introduces new ones. 
Our main result shows that the classic \emph{Online Gradient Descent (GD)} algorithm achieves an optimal $O(\sqrt{T})$ bound on proximal regret simultaneously for all, possibly nonconvex, lower semi-continuous functions, revealing that GD, without modification, minimizes a stronger regret notion than external regret.
This provides a new explanation for the empirically superior performance of gradient descent in online learning and games. 
We further extend our analysis to \emph{Mirror Descent} in the Bregman setting and, for proximal regret with respect to convex functions, to \emph{Optimistic Gradient Descent}, which yields faster convergence in smooth convex games.
\end{abstract}
\thispagestyle{empty}
\clearpage
\thispagestyle{empty}
\tableofcontents
\clearpage
\setcounter{page}{1}

\section{Introduction}
Learning and computation of equilibria have wide-ranging applications in computer science, game theory, economics, and optimization. Since von Neumann’s celebrated minimax theorem~\citep{v1928theorie}, establishing the existence of minimax equilibria in two-player zero-sum games, and Nash’s work~\citep{nash1950equilibrium}, showing the existence of Nash equilibria in multi-player and/or general-sum games, there has been a long line of research studying the statistical and computational complexity of equilibrium computation and learning. These efforts span normal-form games (see e.g.~\citep{von1944theory, lemke1964equilibrium, freund1997decision, fudenberg1998theory, daskalakis2009complexity, chen2009settling, anagnostides2025computational}), extensive-form games (see e.g.~\citep{kuhn1953contributions, zinkevich2007regret, von2008extensive, Farina22:Simple, farina2022kernelized}), Bayesian games (see e.g.~\citep{harsanyi1967games, hartline2015no, peng2024complexity}), the broad family of convex games (see e.g.~\citep{rosen1965existence, stoltz2007learning, gordon2008no, farina2022near, papadimitriou2023computational, daskalakis2025efficient}), as well as stochastic games (see e.g.~\citep{bai2020provable,jin2021v,daskalakis2020independent,daskalakis2023complexity}).

We focus on convex games in this paper, and our central question is when and what types of equilibria can be computed efficiently and/or learned efficiently by interacting agents. This question is fundamental to justify the use of a certain kind of equilibrium as a solution concept to predict the possible outcomes of a complex multi-agent interaction wherein self-interested agents adapt their strategies through learning.

Nash equilibrium is perhaps the most well-known solution concept in game theory and provides a strong form of strategic stability: no player can unilaterally deviate to improve their utility. Moreover, computing a Nash equilibrium is efficient in two-player zero-sum games via linear programming, and simple learning dynamics are known to converge to Nash equilibria even in the last-iterate sense; see e.g.~\citep{freund1997decision, daskalakis2019last, cai2022finite}. However, this strong tractability and learnability do not generalize. Computing a Nash equilibrium is intractable even for two-player bimatrix games~\citep{daskalakis2009complexity, chen2009settling, rubinstein2017settling}, and learning dynamics generally fail to converge to the set of Nash equilibria~\citep{milionis2023impossibility}.

Given the strong intractability results for Nash equilibrium, relaxed notions of equilibrium have been extensively studied. Among them, coarse correlated equilibrium (CCE) and correlated equilibrium (CE)~\citep{aumann1987correlated} emerge as natural solution concepts due to their intimate relationship to no-regret online learning methods~\citep{cesa2006prediction}. Indeed, when every player in some game employs a no-external-regret algorithm to iteratively update their strategy in response to the history of play of the other players, the empirical distribution of strategies converges to the set of CCE; when every player uses a no-swap-regret algorithm, the empirical distribution converges to the set of CE. 

CCE is tractable thanks to the existence of many simple and efficient no-external-regret learning algorithms, such as online gradient descent~\citep{zinkevich2003online}. However, it may be undesirable as a weak notion of equilibrium. For instance, a CCE can exhibit a high price of anarchy (PoA)~\citep{koutsoupias1999worst}, meaning that the worst CCE may yield outcomes substantially less desirable than Nash equilibrium \citep{syrgkanis2013composable, feldman2016correlated,jin2023first}. 

In contrast, CE provides stronger strategic stability and typically leads to more desirable outcomes~\citep{feldman2016correlated, hartline2024regulation, hartline2025regulation}. No-swap-regret learning algorithms~\citep{blum2007external,gordon2008no} further possess advantages absent in no-external-regret algorithms: they are non-exploitable~\citep{deng2019strategizing}, Pareto-optimal~\citep{arunachaleswaran2024pareto}, and imply sequential calibration~\citep{fishelson2025high}. Despite their appeal, efficient no-swap-regret learning algorithms---and consequently, efficient algorithms for computing CE---remain elusive. While simple learning dynamics can converge to CE in normal-form games~\citep{foster1997calibrated,hart2000simple} and efficient algorithms exist for succinct games~\citep{papadimitriou2008computing}, it is unknown how to compute CE efficiently in general convex games. 

Recent breakthroughs have produced general swap-regret minimization algorithms~\citep{dagan2024external,peng2024fast}, yet these exhibit exponential dependence on the admissible regret. Moreover, any swap-regret minimization algorithm must incur exponential dependence on either the problem dimension or the admissible regret, even in extensive-form games~\citep{daskalakis2024lower}.

Towards efficient learning and computation of equilibria in games, we focus on \emph{$\Phi$-regret}~\citep{greenwald2003general}, a unified and flexible notion encompassing both external and swap regret in online learning. 
Given a strategy set $\mathcal{X} \subseteq \mathbb{R}^d$ and a sequence of loss functions $\ell^1, \ldots, \ell^T$, the $\Phi$-regret of a sequence of actions $x^1, \ldots, x^T$ is defined as
\[
\sum_{t=1}^T \ell^t(x^t) \;-\; \min_{\phi \in \Phi} \sum_{t=1}^T \ell^t(\phi(x^t)).
\]
Here, $\Phi$ is a set of strategy modifications, and each modification $\phi: \mathcal{X} \to \mathcal{X}$ is an endomorphism. 
Both external and swap regret can be viewed as special cases of $\Phi$-regret: external regret corresponds to the case where $\Phi$ contains all constant functions, while swap regret corresponds to the case when $\Phi$ contains all functions from $\mathcal{X}$ to itself. 

Analogous to the relationship between CCE and CE, when every player $i$ employs a no-$\Phi$-regret learning algorithm, the empirical distribution of strategies played converges to the set of $\Phi$-equilibria. 
A \emph{$\Phi$-equilibrium} is a distribution over joint strategy profiles such that no player can improve their expected utility by applying any strategy modification $\phi \in \Phi$.

$\Phi$-regret provides a fine-grained framework that interpolates between external and swap regret. 
While minimizing swap regret efficiently is known to be hard in general, 
a central open question is to characterize which notions of $\Phi$-regret \emph{can} be minimized efficiently. 
In this paper, we call an algorithm \emph{efficient} if it achieves a time-average regret of at most $\varepsilon$ 
within $T = \poly(\nicefrac{1}{\varepsilon}, d)$ rounds, where recall that $d$ is the dimension.
This motivates the following questions, originally raised by~\citet{daskalakis2025efficient}:

\medskip\noindent
    \hspace{.1in}
    \begin{minipage}
    {0.95\textwidth}\textbf{Motivating Questions:} Are there notions of $\Phi$-regret that are stronger than external regret and can be minimized efficiently against an \emph{adversarial} sequence of convex losses? 
Relatedly, are there notions of $\Phi$-equilibrium stronger than CCE that can be computed efficiently in general convex games?
\end{minipage}
\medskip
\noindent 

\subsection{Our Results}
We provide a positive answer to this motivating question and introduce \emph{proximal regret}, a new tractable notion of $\Phi$-regret that is strictly stronger than external regret. The proximal regret builds on the \emph{proximal operator}, a fundamental notion in optimization theory~\citep{parikh2014proximal}. Given any proper lower semi-continuous function $f: \+X \rightarrow (-\infty,+\infty]$, which need not be convex, the proximal operator $\prox_f: \+X \rightarrow \+X$ is a mapping defined as\footnote{When $f$ is not convex, the minimizer may not be unique. In that case, $\prox_f(x)$ denotes an arbitrary minimizer, and all of our guarantees hold for every such choice. See \Cref{sec:pre} for formal definitions.}
\begin{align*}
    \prox_f(x) \in \argmin_{x' \in \+X} \left\{ f(x') + \frac{1}{2}\InNorms{x' - x}^2\right\}.
\end{align*}
For a sequence of actions $x^1, \ldots, x^T$, the proximal regret with respect to $f$ against a sequence of convex loss functions $\ell^1, \ldots, \ell^T$ is defined as 
\begin{align*}
    \reg_f^T:= \sum_{t=1}^T \ell^t(x^t) - \sum_{t=1}^T \ell^t(\prox_f(x^t)).
\end{align*}
For a set of proper lower semi-continuous functions $\+F$, the proximal regret with respect to $\+F$ is $\reg_{\+F}^T = \sup_{f \in \+F} \reg^T_f$.

In what follows, we illustrate several special cases of proximal regret $\reg_{\+F}^T$ arising from different choices of the function class $\+F$. Although our framework allows any proper lower semi-continuous function, proximal regret already exhibits a rich behavior for certain natural subclasses.
\begin{itemize}
    \item \textbf{External Regret and Beyond.} 
We first show that proximal regret captures external regret (see \Cref{sec:pre} for a formal definition). 
To recover external regret, we can choose the set of indicator functions $\+F_{\ext} = \{\-I_{\{x\}} : x \in \+X\}$.\footnote{An indicator function $\mathbb{I}_{\{x\}}$ is such that $\mathbb{I}_{\{x\}}(x) = 0$ and $\mathbb{I}_{\{x\}}(x') = +\infty$ for $x' \ne x$.}
Since $\prox_{\-I_{\{x\}}}(\cdot) = x$ is a constant function, the proximal regret $\reg_{\+F_{\ext}}$ is equivalent to the external regret.
More generally, we can take $\+F = \{\-I_S : S \subseteq \+X \text{ is closed}\}$, in which case the deviation $\prox_{\-I_S}(\cdot)$ maps every point to a nearest point in $S$; when $S$ is convex, this is the $\ell_2$ projection onto $S$.
To the best of our knowledge, the corresponding notion of proximal regret is novel and has not been explored in the literature.
    \item \textbf{Non-negativity.} While external regret may be negative, the proximal regret is non-negative whenever we allow a constant function $f$. This is because $\prox_f$ is the identity mapping for constant $f$.
    \item \textbf{Gradient Equilibrium.} 
A recent work~\citep{angelopoulos2025gradient} proposes a property for online learning called \emph{gradient equilibrium}, which requires that the averaged gradient norm vanishes, i.e., 
\[
\InNorms{\tfrac{1}{T}\sum_{t=1}^T \nabla \ell^t(x^t)} = o(1).
\]
The gradient equilibrium property is incomparable with the no-external-regret property and has interesting implications in online conformal prediction~\citep{angelopoulos2025gradient,ramalingam2025relationship}. 
We note that gradient equilibrium is implied by proximal regret with respect to linear functions. 
Specifically, let $\+F = \{ f_v(x) = \InAngles{v, x} : \InNorms{v} = 1 \}$. 
Then $\prox_{f_v}(x) = \Pi_{\+X}[x - v]$ for all $x$. 
This special case of proximal regret, $\reg_{\+F}^T$, has been studied under the names of \emph{projection-based regret}~\citep{cai2024tractable} and \emph{no-move regret}~\citep{angelopoulos2025gradient}. 
In the setting of online conformal prediction, where $\+X = \-R^d$ is unconstrained, sublinear \emph{linearized} proximal regret, $\reg_{\+F}^T = o(T)$, implies the gradient equilibrium property.\footnote{The linearized regret is defined as $\sum_{t=1}^T \InAngles{\nabla \ell^t(x^t), x^t - \prox_f(x^t)}$, which is slightly stronger than the standard regret when the losses $\{\ell^t\}_{t \in [T]}$ are convex. All of our results hold for linearized regret.}
    \item \textbf{Symmetric Linear Swap Regret.} 
By choosing $\+F$ to be quadratic functions that are not necessarily convex, the strategy modifications induced by the proximal operator include all affine endomorphisms $\phi(x) = Ax + b$ with symmetric positive definite $A$, and, up to a simple reduction, cover all \emph{symmetric} affine endomorphisms, i.e., mappings of the form $\phi(x) = A x + b$ where $A$ is symmetric (see \Cref{sec:symmetric linear swap regret} for details).
We refer to the corresponding $\Phi$-regret as \emph{symmetric linear swap regret}, and the corresponding $\Phi$-equilibrium as \emph{symmetric linear correlated equilibrium}. 
In the special case of normal-form games where strategies lie in the simplex $\+X = \Delta^d$, symmetric linear correlated equilibrium coincides with a refined notion of CCE called \emph{semicoarse correlated equilibrium}~\citep{ahunbay2025semicoarse}. 
In certain first-price auctions, while CCE may yield low social welfare, every symmetric linear correlated equilibrium achieves optimal social welfare~\citep{ahunbay2025semicoarse}.
\end{itemize}

From the above examples, we see that proximal regret forms a general class of $\Phi$-regret that is stronger than external regret and admits non-trivial guarantees in online learning and games. 
Given the generality of proximal regret---the function $f$ may be nonconvex and even discontinuous, and its proximal operator may be discontinuous---it is a priori unclear how to minimize it efficiently.
Surprisingly, we show that the simple yet fundamental no-external-regret algorithm, \hyperref[GD]{Online Gradient Descent (GD)}, achieves sublinear proximal regret, with the same guarantee as for external regret.

\medskip
\noindent
\hspace{.1in}
\begin{minipage}{0.95\textwidth}
\textbf{Main Result.} 
We show that \hyperref[GD]{Online Gradient Descent (GD)} minimizes proximal regret in online convex optimization. 
Specifically, over any compact convex strategy set, \hyperref[GD]{GD} guarantees $\reg^T_f = O(\sqrt{T})$ proximal regret simultaneously for all proper lower semi-continuous functions $f$. Moreover, the bound is independent of $f$ and coincides with the classical external regret bound of \hyperref[GD]{GD} (see \Cref{theorem:GD proximal regret} for the formal statement).
\end{minipage}
\medskip
\noindent

Our result significantly extends the classic guarantee of \hyperref[GD]{Online Gradient Descent}~\citep{zinkevich2003online}, showing that \hyperref[GD]{GD} minimizes not only external regret but also a stronger notion of proximal regret. 
In particular, we demonstrate that the sequence produced by \hyperref[GD]{GD} is not merely no-external-regret, but also satisfies additional constraints induced by the proximal operators of all proper lower semi-continuous functions.
Moreover, our result implies that \hyperref[GD]{GD} dynamics in convex games converge to a refined subset of coarse correlated equilibria, which we refer to as \emph{proximal correlated equilibria (PCE)}.

\paragraph{Analysis overview.} Our analysis is extremely simple and can be seen as a generalization of the classic analysis of online gradient descent by~\citet{zinkevich2003online}. The key observation is that the \emph{Moreau envelope} $M_f(x) = \min_{x' \in \+X}\{ f(x') + \frac{1}{2}\InNorms{x' - x}^2\}$ is a natural potential function for proximal regret: the deviation $\prox_f(x^t)$ is a feasible point of the minimization problem defining $M_f(x^{t+1})$, so each step of \hyperref[GD]{GD} decreases the potential $M_f$ by $\eta_t$ times the instantaneous linearized proximal regret, up to the usual $\frac{\eta_t^2}{2}\InNorms{\nabla \ell^t(x^t)}^2$ term. The regret bound then follows by telescoping.


\subsection{Extensions and Discussions} 

\paragraph{Extension to the Bregman Setup and Mirror Descent.}  
Our analysis in the Euclidean setup naturally generalizes to the Bregman setup, where the proximal operator is defined using a general Bregman divergence instead of the Euclidean distance. 
For different choices of Bregman divergence, the corresponding Bregman proximal operator induces a new tractable class of $\Phi$-regret, which we call \emph{Bregman proximal regret}. 
In \Cref{sec:bregman proximal}, we extend our analysis of \hyperref[GD]{GD} and show that the \hyperref[MD]{Mirror Descent} algorithm~\citep{nemirovskij1983problem} minimizes Bregman proximal regret with respect to all proper lower semi-continuous functions over compact convex sets, with the same guarantee as its classical external regret bound.

\paragraph{Fast Convergence and Social Regret in Convex Games.} 
Recall that \Cref{theorem:GD proximal regret} establishes an optimal $O(\sqrt{T})$ proximal regret bound for \hyperref[GD]{GD} against an adversarial sequence of loss functions. 
When every player employs \hyperref[GD]{GD} in a convex game, the empirical distribution of strategies converges to PCE at a rate of $O(T^{-1/2})$. 
However, since each player does not face a fully adversarial environment, smaller individual regret and faster convergence to PCE are possible in smooth convex games. 

Starting from the pioneering work of~\citet{daskalakis2011near}, a long line of research on learning in games~\citep{daskalakis2011near,rakhlin2013optimization,syrgkanis2015fast,chen2020hedging,anagnostides2022near-optimal,anagnostides2022uncoupled,piliouras2022beyond,farina2022near,cai2024near,soleymani2025faster} has established accelerated convergence results in normal-form and Markov games. 
Many of these results employ the optimism technique, which stabilizes the learning process. 
In the context of general convex games, the current state of the art is $O(\log T)$ individual external regret and $O(1)$ social external regret. 
For the stronger notion of linear swap regret, no bound better than $O(\sqrt{T})$ is known.

In \Cref{sec:OG game setting}, we extend our analysis of \hyperref[GD]{GD} to Optimistic Gradient Descent (\ref{OG}), yielding improved instance-dependent regret bounds (\Cref{thm:adversaril regret of OG}) and faster convergence in the game setting (\Cref{thm:game regret of OG}). Specifically, for general convex games, we obtain $\reg_{\+F_{\convex}}^T=O(T^{1/4})$ individual proximal regret bound for convex functions $\+F_{\convex}$, and $O(1)$ social proximal regret $\reg_{\+F_{\strc}}^T$ for strongly convex functions $\+F_{\strc}$. Both results are new and complement existing results. We also discuss a possible path to obtain $O(1)$ individual external regret in convex games using proximal regret.

\paragraph{Proximal Regret Demystifies Some Successes of \hyperref[GD]{GD}.} 
As a fundamental algorithm in online learning and games, \hyperref[GD]{GD} and its variants find applications in a wide range of settings. 
Despite its simplicity, \hyperref[GD]{GD} often exhibits superior performance compared to generic no-external-regret algorithms. 
Our results offer a new perspective on why \hyperref[GD]{GD} is special: it minimizes the stronger notion of proximal regret, and its dynamics in convex games converge to proximal correlated equilibria, a strict subset of CCE.

In online conformal prediction, the goal is to guarantee marginal coverage in adversarial environments. 
While running \hyperref[GD]{GD} against the pinball loss ensures marginal coverage~\citep{gibbs2021adaptive}, recent concurrent works~\citep{angelopoulos2025gradient,ramalingam2025relationship} show that achieving no-external-regret against the pinball loss alone is insufficient for this guarantee. 
By adapting the analyses in~\citep{angelopoulos2025gradient,ramalingam2025relationship}, we can show that \hyperref[GD]{GD} achieves sublinear proximal regret $\reg_{\+F}^T$ (where $\+F$ contains bounded linear functions) in the online conformal prediction setting. 
As discussed above, sublinear proximal regret implies the gradient equilibrium property, which in turn guarantees marginal coverage~\citep{angelopoulos2025gradient}.

In fixed-value first-price auctions with at least three highest-value bidders, it is known that the price of anarchy (PoA) of CCE is $> 1$ while the PoA of CE is $1$~\citep{feldman2016correlated}. Recent work~\citep{ahunbay2025semicoarse} shows that the PoA of semicoarse correlated equilibria is also $1$ and when every player runs \hyperref[GD]{GD} with the same step size, the empirical distribution converges to semicoarse correlated equilibria. As discussed above, proximal correlated equilibria $\subseteq$ symmetric linear correlated equilibria $=$ semicoarse correlated equilibria in normal-form games. Thus our results on the adversarial proximal regret of \hyperref[GD]{GD} directly imply the social welfare guarantee of \hyperref[GD]{GD} dynamics and do not require each player to use the same step size. 

\subsection{Additional Related Works}
We note that the work of~\citet{daskalakis2025efficient} also provides one positive answer to the motivating questions. They consider \emph{linear} swap regret, a special case of $\Phi$-regret where $\Phi$ contains all (affine) linear endomorphisms over $\mathcal{X}$, that is, mappings of the form $\phi(x) = A x + b$. They present an efficient algorithm for minimizing linear swap regret and a polynomial-time algorithm for computing the corresponding notion of linear CE. Subsequently,~\citet{zhang2025learning} extend these results to settings where $\Phi$ contains low-dimensional functions, including low-degree polynomials. However, the tractability of $\Phi$-regret minimization beyond these settings remains unknown. Moreover, the no-regret algorithms proposed in~\citep{daskalakis2025efficient,zhang2025learning} are inspired by the \emph{Ellipsoid Against Hope} (EAH) framework~\citep{papadimitriou2008computing}, relying on sophisticated subroutines that may not be practical for large-scale settings. Very recently, \citet{anagnostides2026swap} give more efficient algorithms based on Blackwell's approachability for linear and lower-degree polynomial swap-regret minimization. But their algorithm requires solving a zero-sum game each step and is arguably more complicated than gradient-based algorithms.

\citet{gordon2008no} proposes a general framework for no-$\Phi$-regret algorithms, which requires (i) a no-external-regret algorithm over the set of strategy modifications $\Phi$, which is in general challenging to obtain, and (ii) a fixed-point oracle for mappings in $\Phi$. 
The work of~\citep{zhang2024efficient} generalizes this framework and shows that an ``expected fixed-point'' oracle suffices, which can be implemented efficiently. 
The GGM framework is powerful and underlies most known constructions of no-$\Phi$-regret algorithms in the literature. 
In contrast, our results on proximal regret and \hyperref[GD]{GD} show that a simple algorithm can minimize $\Phi$-regret---stronger than external regret---without requiring the strong primitives used in GGM.

The work of~\citet{cai2024tractable} studied the guarantee of \hyperref[GD]{GD} for $\Phi$-regret minimization and show that \hyperref[GD]{GD} minimizes a special class of $\Phi$-regret called projection-based regret (as we discussed above) and interpolation-based regret in online convex optimization.
Building on~\citep{cai2024tractable}, the early version of \citet{ahunbay2024first} proposed a generalized set of strategy modifications based on gradient fields, which include both projection-based and interpolation-based regret as special cases. 
His work encouraged us to consider a broad family of strategy modifications based on proximal operators—a family much more general than those in~\citep{cai2024tractable}—which we present here. 
While the two approaches are different and largely incomparable, the key distinction is that the early version of \citet{ahunbay2024first} studies different notions of equilibrium and focuses on providing inequality constraints for gradient flow in games, and does not address $\Phi$-regret in the adversarial online learning setting or the approximation of standard $\Phi$-equilibria, which are the main focus of our work. 

Concurrent to our results, a later version of \citet{ahunbay2024first} presents an adversarial proximal regret guarantee for \hyperref[GD]{GD} but requires additional regularity assumptions on the action sets. Building on our results, Ahunbay extended his original analysis and obtained an adversarial proximal regret guarantee for \hyperref[GD]{GD} for general convex sets in the latest version of~\citep{ahunbay2024first}. Moreover, while our analysis naturally generalizes to \hyperref[MD]{Mirror Descent (MD)} and the Bregman setting, his analysis does not fully extend to the Bregman case but requires the distance-generating function $\phi$ to be ``steep''. See \Cref{app:comparison} for a more detailed discussion on \citep{cai2024tractable,ahunbay2024first}.

\section{Preliminary}\label{sec:pre}
We denote $\-R$ the set of real numbers and $\bR = \-R \cup \{+\infty\}$. A (extended) real-valued function $f: \+X \rightarrow \bR$ is proper if it is not identically equal to $+\infty$; $f$ is lower semi-continuous if $f(x) \le \liminf_{y \rightarrow x} f(y)$ for all $x \in \+X$; $f$ is $L$-smooth if its gradient $\nabla f$ is $L$-Lipschitz continuous. For a convex function $f$, its subdifferential $\partial f(y)$ at $y$ is the set of all vectors $v$ such that $f(x) \ge f(y) + \InAngles{v, x - y}$ for all $x \in \+X$. A function $f$ is $\alpha$-strongly convex if $f - \frac{\alpha}{2}\InNorms{\cdot}^2$ is convex, which implies $f(x) \ge f(y)+ \InAngles{v, x- y} + \frac{\alpha}{2}\InNorms{x - y}^2$ for all $x, y \in \+X$ and $v \in \partial f(y)$; $f$ is $\rho$-weakly convex for $\rho \ge 0$ if $f + \frac{\rho}{2}\InNorms{\cdot}^2$ is convex. Every convex function is $0$-weakly convex, and every $L$-smooth function is $L$-weakly convex.

We denote the class of all proper and lower semi-continuous functions as $\+F_{\lsc}(\+X)$, and its subclass of convex functions as $\+F_{\convex}(\+X)$. Throughout the paper, we may omit $\+X$ and write $\+F_{\lsc}$ and $\+F_{\convex}$ when the context is clear. For any function in $\+F_{\lsc}$, we can define the proximal operator, which has been extensively studied in the optimization literature. We refer the readers to the book \citep{parikh2014proximal} for a comprehensive review of proximal operators of convex functions, and to the book \citep{rockafellar1998variational} for proximal operators of nonconvex functions.

\begin{definition}[Proximal Operator]\label{dfn:proximal operator}
    A \emph{proximal operator} associated with a proper, lower semi-continuous function $f : \+X \rightarrow \bR$ is any mapping $\prox_f: \+X \rightarrow \+X$ such that for all $x \in \+X$,
    \begin{align*}
    \prox_{f}(x) \in \argmin_{x' \in \+X} \left\{ f(x') + \frac{1}{2}\InNorms{x'-x}^2 \right\}.
    \end{align*}
\end{definition}
The minimum in \Cref{dfn:proximal operator} is attained for every $x \in \+X$ whenever $\+X$ is compact, since the objective is proper and lower semi-continuous; when $\+X$ is closed, it is also attained whenever $f$ is bounded from below or $\rho$-weakly convex with $\rho < 1$ (e.g., convex). Throughout the paper, whenever we refer to $\prox_f$, we assume that the minimum is attained for every $x \in \+X$. If $f$ is $\rho$-weakly convex with $\rho < 1$ (in particular, if $f$ is convex), then the objective is strongly convex, so the minimizer is unique and the proximal operator $\prox_f$ is uniquely determined. In general, however, the minimizer may not be unique and $\prox_f$ may be discontinuous. For example, for the indicator function $\-I_S$ of a closed set $S \subseteq \+X$, $\prox_{\-I_S}(x)$ is a nearest point to $x$ in $S$, which may not be unique when $S$ is not convex. All of our results hold for every choice of the proximal operator. When $f$ is convex, by the first-order optimality condition of the optimization problem in the definition of the proximal operator, we have the following fact.
\begin{fact}\label{fact:proximal operator}
    Let $f \in \+F_{\convex}(\+X)$ and $p = \prox_f(x)$. Then there exists a subgradient $v \in \partial f(p)$ such that $\InAngles{x- v - p, x'-p} \le 0$ for all $x' \in \+X$.
\end{fact}

\paragraph{Moreau Envelope} Our analysis of proximal regret relies on the Moreau envelope~\citep{moreau1965proximite,rockafellar1998variational}, which is the optimal value of the optimization problem in the definition of the proximal operator.
\begin{definition}[Moreau Envelope]\label{dfn:moreau envelope}
    The \emph{Moreau envelope} of a proper, lower semi-continuous function $f : \+X \rightarrow \bR$ is
    \begin{align*}
        M_f(x) := \min_{x' \in \+X} \left\{ f(x') + \frac{1}{2}\InNorms{x'-x}^2 \right\}, \quad \forall x \in \+X.
    \end{align*}
\end{definition}
The following two properties of the Moreau envelope are the key to our analysis.
\begin{lemma}\label{lemma:moreau envelope}
    Let $f \in \+F_{\lsc}(\+X)$ and let $\prox_f$ be a proximal operator associated with $f$. Then
    \begin{enumerate}
        \item for all $x, y \in \+X$, $M_f(y) \le f(\prox_f(x)) + \frac{1}{2}\InNorms{\prox_f(x) - y}^2$, where the equality holds if $y = x$;
        \item if $\+X$ is compact with diameter $D$, then $\min_{x' \in \+X} f(x') \le  M_f(x) \le \min_{x' \in \+X} f(x') + \frac{D^2}{2}$ for all $x \in \+X$.
    \end{enumerate}
\end{lemma}
\begin{proof}
    Item 1 holds since $\prox_f(x) \in \+X$ is a feasible point of the minimization problem defining $M_f(y)$, and it is a minimizer when $y = x$. For item 2, since $f$ is proper and lower semi-continuous, it attains its minimum over the compact set $\+X$ at some point $z \in \+X$. Then by item 1 and the definition of $M_f$,
    \begin{align*}
        f(z) \le f(\prox_f(x)) \le f(\prox_f(x)) + \frac{1}{2}\InNorms{\prox_f(x) - x}^2 = M_f(x) \le f(z) + \frac{1}{2}\InNorms{z - x}^2 \le f(z) + \frac{D^2}{2}.
    \end{align*}
    This completes the proof.
\end{proof}

\paragraph{Online Learning and Regret} In the online learning setup, every day $t \ge 1$, the learner chooses an action/strategy $x^t \in \+X$ while the adversary picks a loss function $\ell^t : \+X \rightarrow \-R$. The learner then incurs loss $\ell^t(x^t)$ and receives gradient feedback $\nabla \ell^t(x^t)$. An online learning algorithm decides the actions $x^t$ based on information received in the previous $t-1$ days. The goal of an online learning algorithm is to minimize the $\Phi$-regret, where $\Phi = \Phi(\+X)$ is a set of strategy modifications $\phi: \+X \rightarrow \+X$:
\begin{align*}
    \reg_\Phi^T:= \sup_{\phi \in \Phi} \InParentheses{ \sum_{t=1}^T \ell^t(x^t) - \ell^t(\phi(x^t))}.
\end{align*}
The $\Phi$-regret quantifies the difference between the cumulative loss of the algorithm and that obtained by applying the best strategy modification. The $\Phi$-regret framework is general and covers several interesting notions. In particular, when we let $\Phi_{\mathrm{ext}}$ be all the constant functions over $\+X$, the resulting $\Phi_{\mathrm{ext}}$-regret is the \emph{external regret}:
\begin{align*}
    \reg_{\mathrm{ext}}^T := \max_{x \in \+X}\InParentheses{ \sum_{t=1}^T \ell^t(x^t) - \ell^t(x)}.
\end{align*}
If we let $\Phi_{\mathrm{swap}}$ be the set of all possible strategy modifications over $\+X$, the resulting $\Phi_{\mathrm{swap}}$-regret is called the swap regret.\footnote{The usual notion of swap regret is for the simplex $\+X = \Delta^d$ and allows all \emph{linear} strategy modifications. Here, we consider the stronger notion that works for general convex sets and all possible strategy modifications, which is also called full swap regret~\citep{fishelson2025full}.} 

\paragraph{Smooth Convex Games and Equilibria} A smooth convex game $\+G = \{[n], \{\+X_i\}_{i \in [n]}, \{u_i\}_{i \in [n]}\}$ has $n$ players $[n] = \{1, 2,\ldots, n\}$. Each player $i \in [n]$ has a compact convex strategy set $\+X_i \subseteq \-R^{d_i}$ and a continuously differentiable and concave utility function $u_i: \times_{i \in [n]} \+X_i \rightarrow [0,1]$.\footnote{Note that the utility is concave. We call the game a convex game just for notation consistency} We assume each $u_i$ is $G$-Lipschitz, i.e, $\InNorms{\nabla u_i(\cdot)}\le G$,  and $L$-smooth, i.e, $\InNorms{\nabla_{x_i} u_i(x) - \nabla_{x_i} u_i(x')} \le L\InNorms{x - x'}$ for all $x, x' \in \times_{i \in [n]} \+X_i$. Smooth convex games cover finite normal-form games, Bayesian games, and extensive-form games. 
\begin{definition}[$\Phi$-equilibrium]
    A $\varepsilon$-approximate $\Phi$-equilibrium in a smooth convex game with strategy sets $\+X_1, \ldots, \+X_n$ is a joint distribution $\sigma \in \Delta(\times_{i \in [n]} \+X_i)$ such that every player $i \in [n]$ and any strategy modifications $\phi \in \Phi(\+X_i)$, 
    \begin{align*}
        \-E_{x \sim \sigma}\InBrackets{u_i(x)} \ge \-E_{x \sim \sigma}\InBrackets{u_i(\phi(x_i), x_{-i})} - \varepsilon.
    \end{align*}
    We say $\sigma$ is a $\Phi$-equilibrium if the above inequalities holds with $\varepsilon=0$.
\end{definition}
It is well known that when every player employs an online learning algorithm in a smooth convex game, the empirical distribution of strategies played converges to the set of $(\max_i\{\reg^T_{\Phi(\+X_i)}/T\})$-approximate $\Phi$-equilibria~\citep{greenwald2003general}.

\section{Proximal Regret: A Tractable Notion Between External and Swap Regret}
\label{sec:proximal regret}
In this section, we introduce a class of $\Phi$-regret called \emph{proximal regret} parameterized by proper lower semi-continuous functions. The main result of this section shows that the classic \hyperref[GD]{Online Gradient Descent} algorithm~\citep{zinkevich2003online} minimizes proximal regret. Since external regret is an instantiation of proximal regret, our result shows that the proximal regret is an efficient regret notion that lies between external regret and swap regret.

\paragraph{Proximal Regret} The strategy modification in proximal regret is defined using the proximal operator.
\begin{definition}[Proximal Regret]\label{dfn:proximal regret}
    For a sequence of losses $\ell^1, \ldots, \ell^T$, the \emph{proximal regret} associated with $f\in \+F_{\lsc}(\+X)$ and a proximal operator $\prox_f$ of $f$ is
    \begin{align*}
    \reg^T_f:= \sum_{t=1}^T \ell^t(x^t) - \ell^t(\prox_f(x^t)).
\end{align*}
The proximal regret for a family of functions $\+F \subseteq \+F_{\lsc}(\+X)$ is defined as $\sup_{f \in \+F} \reg^T_f$, where the supremum is also over all proximal operators of each $f \in \+F$. We say an algorithm is \emph{no-proximal regret} w.r.t. $\+F$ if $\sup_{f \in \+F} \reg^T_f = o(T)$.
\end{definition}
In other words, the proximal regret w.r.t. $\+F$ is the $\Phi$-regret with $\Phi = \{\prox_f : f \in \+F\}$, where $\prox_f$ ranges over all proximal operators of $f$. We denote by $\Phiprox := \{\prox_f : f \in \+F_{\lsc}(\+X)\}$ the set of all proximal operators over $\+X$, which we characterize in \Cref{sec:characterization}. We call a $\Phi$-equilibrium with $\Phi(\+X_i) = \Phi^{\+X_i}_{\prox}$ for every player $i$ a \emph{proximal correlated equilibrium (PCE)}.
As we have discussed in the introduction, if we choose $\+F_{\mathrm{ext}}(\+X) = \{ \-I _{\{x\}}: x\in \+X\}$ to be all the indicator functions for one point in $\+X$, then the corresponding proximal regret $\max_{f \in \+F_{\mathrm{ext}}(\+X)}\reg_f^T = \reg^T_{\mathrm{ext}}$ is exactly the external regret. Thus, external regret is a special case of proximal regret. 

\subsection{Online Gradient Descent Minimizes Proximal Regret}

\begin{algorithm}[!ht]\label{GD}
    \KwIn{strategy space $\+X$, non-increasing step sizes $\{\eta_t\}$} 
    \KwOut{strategies $\{x^t\}$}
    \caption{Online Gradient Descent (GD)}
    Initialize $x^1 \in \+X$ arbitrarily\\
    \For{$t = 1,2, \ldots,$}{
    play $x^t$ and receive feedback $\nabla \ell^t(x^t)$.\footnotemark\\
    update $x^{t+1} = \Pi_\+X\InBrackets{x^t - \eta_t \nabla \ell^t(x^t)}$.
    }
\end{algorithm}\footnotetext{When the function is non-differentiable, the algorithm receives a subgradient in $\partial \ell^t(x^t)$. In fact, our results work for any vector feedback $g^t \in \mathbb{R}^d$ and the linearized regret $\sum_{t=1}^T \InAngles{g^t, x^t - \prox_f(x^t)}$.}

We show a surprising result that over any compact convex set $\+X$, \hyperref[GD]{Online Gradient Descent (GD)} achieves $\reg_f^T = O(\sqrt{T})$ \emph{simultaneously} for all $f \in \+F_{\lsc}(\+X)$ and all of their proximal operators, without any modification. Remarkably, the bound does not depend on $f$ and is identical to the classical external regret bound of \hyperref[GD]{GD}. Our result significantly extends the classical result of \hyperref[GD]{GD}~\citep{zinkevich2003online} by showing that \hyperref[GD]{GD} not only minimizes external regret but simultaneously the more general proximal regret. It shows that the sequence produced by \hyperref[GD]{GD} is not only a no-external-regret sequence but has additional features. It also implies that the \hyperref[GD]{GD} dynamics in convex games converge to a refined subset of coarse correlated equilibria, which we refer to as \emph{proximal correlated equilibria (PCE)}. We show more applications of proximal regret and proximal correlated equilibrium in later sections.

\begin{theorem}[Online Gradient Descent Minimizes Proximal Regret]\label{theorem:GD proximal regret}
    Let $\+X \subseteq \-R^d$ be a closed convex set and $\{\ell^t: \+X \rightarrow \-R\}_{t \in [T]}$ be a sequence of convex loss functions. Let $\{x^t\}_{t \in [T+1]}$ be the iterates generated by \hyperref[GD]{Online Gradient Descent (GD)} with non-increasing step sizes $\{\eta_t\}$. Then for any $f \in \+F_{\lsc}(\+X)$ and any proximal operator $\prox_f$ associated with $f$, denoting $p^t = \prox_f(x^t)$, we have
    \begin{align*}
        &\sum_{t=1}^T \ell^t(x^t) -\sum_{t=1}^T \ell^t(p^t) \le \sum_{t=1}^T \InAngles{\nabla \ell^t(x^t), x^t - p^t} \le \sum_{t=1}^T \frac{M_f(x^t) - M_f(x^{t+1})}{\eta_t} + \sum_{t=1}^T\frac{\eta_t}{2}\InNorms{\nabla \ell^t(x^t)}^2,
    \end{align*}
    where $M_f$ is the Moreau envelope of $f$ (\Cref{dfn:moreau envelope}). In particular,
    if $\+X$ is compact with diameter $D$, then simultaneously for all $f \in \+F_{\lsc}(\+X)$ and all proximal operators $\prox_f$,
        \begin{align*}
            \sum_{t=1}^T \ell^t(x^t) -\sum_{t=1}^T \ell^t(\prox_f(x^t)) \le \sum_{t=1}^T \InAngles{\nabla \ell^t(x^t), x^t - \prox_f(x^t)} \le \frac{D^2}{2\eta_T} + \sum_{t=1}^T\frac{\eta_t}{2}\InNorms{\nabla \ell^t(x^t)}^2.
        \end{align*}
\end{theorem}
\begin{remark}\label{remark:GD}
    Suppose $\+X$ is compact with diameter $D$ and the loss functions are $G$-Lipschitz, so that $\sum_{t=1}^T\frac{\eta_t}{2}\InNorms{\nabla \ell^t(x^t)}^2 \le \frac{G^2}{2} \sum_{t=1}^T \eta_t$. If we choose the decreasing step size $\eta_t = \frac{1}{\sqrt{t}}$, then for all $f \in \+F_{\lsc}(\+X)$ and $T \ge 1$, we get
    \begin{align*}
        \reg_f^T \le \InParentheses{\frac{D^2}{2} + G^2} \cdot \sqrt{T}.
    \end{align*}
    If we know $D$ and $G$ and choose the fixed step size $\eta = \frac{D}{G\sqrt{T}}$, then
    \begin{align*}
        \reg_f^T \le GD\sqrt{T}.
    \end{align*}
    These bounds coincide with the classical bounds for the external regret of \hyperref[GD]{GD}~\citep{zinkevich2003online,orabona2019modern}, which is the special case where $f$ is an indicator function $\-I_{\{x\}}$. In other words, \hyperref[GD]{GD} minimizes proximal regret with respect to all proper lower semi-continuous functions at no additional cost. Since the bounds are uniform over $f$, when every player employs \hyperref[GD]{GD} in a smooth convex game, the empirical distribution of play converges to the set of PCE at a rate of $O(T^{-1/2})$.

    When $\+X$ is unbounded, some restriction on $f$ is necessary: for $\+X = [0, +\infty)$ and the convex function $f(x) = -x$, we have $\prox_f(x) = x+1$, so every algorithm suffers proximal regret $\reg_f^T = T$ against the losses $\ell^t(x) = -x$. Nevertheless, with a fixed step size $\eta_t = \eta$, the first sum in \Cref{theorem:GD proximal regret} telescopes to $\frac{M_f(x^1) - M_f(x^{T+1})}{\eta}$, which still gives an $O(\sqrt{T})$ bound with $\eta = \frac{1}{\sqrt{T}}$, with a constant depending on $f$, whenever $f$ is bounded from below: since $M_f(x^{T+1}) \ge \inf_{x \in \+X} f(x)$, we have $M_f(x^1) - M_f(x^{T+1}) \le M_f(x^1) - \inf_{x \in \+X} f(x)$. In particular, if $f$ attains its minimum at some $z \in \+X$, then $M_f(x^1) - M_f(x^{T+1}) \le \frac{1}{2}\InNorms{x^1 - z}^2$, which matches the external regret bound with comparator $z$.
\end{remark}
\begin{remark}
    Our results also generalize to the Bregman setup where the proximal operator is defined using general Bregman divergence instead of the Euclidean distance. In that case, we show that \hyperref[MD]{Mirror Descent} minimizes the Bregman proximal regret.  See \Cref{sec:bregman proximal} for details.
\end{remark}

We note that proximal regret is essentially a form of \emph{dynamic regret}, where we compete not with a fixed action, but a changing sequence of actions $p^t = \prox_f(x^t)$. The classic analysis of \hyperref[GD]{GD} tracks the distance $\frac{1}{2}\InNorms{x^t - p}^2$ to a fixed comparator $p$. If we follow this analysis, we will have terms $\frac{1}{2}\InNorms{x^{t+1} - p^{t+1}}^2 - \frac{1}{2}\InNorms{x^{t+1} - p^t}^2$ that do not telescope and in the worst case give $\Omega(T)$ regret. The key observation that leads to $O(\sqrt{T})$ regret is that the Moreau envelope $M_f$ is a natural potential function for proximal regret: since the deviation $p^t$ is a feasible point of the minimization problem defining $M_f(x^{t+1})$ (\Cref{lemma:moreau envelope}), the potential $M_f(x^t)$ telescopes even though the comparators $p^t$ change over time. Notably, the argument only uses the global optimality of proximal points, and thus requires no convexity of $f$.

\begin{proof}[Proof of \Cref{theorem:GD proximal regret}]
    Fix any $f \in \+F_{\lsc}(\+X)$ and any proximal operator $\prox_f$ associated with $f$. We denote $p^t = \prox_f(x^t)$ and $g^t = \nabla \ell^t(x^t)$ for $t \in [T]$. Let us fix any $t \in [T]$. Since $p^t \in \+X$, by item 1 of \Cref{lemma:moreau envelope} and the non-expansiveness of the projection $\Pi_\+X$, we have
    \begin{align*}
        M_f(x^{t+1}) &\le f(p^t) + \frac{1}{2}\InNorms{x^{t+1} - p^t}^2 \\
        &= f(p^t) + \frac{1}{2}\InNorms{\Pi_\+X\InBrackets{x^t - \eta_t g^t} - \Pi_\+X\InBrackets{p^t}}^2 \\
        &\le f(p^t) + \frac{1}{2}\InNorms{x^t - \eta_t g^t - p^t}^2 \\
        &= f(p^t) + \frac{1}{2}\InNorms{x^t - p^t}^2 - \eta_t \InAngles{g^t, x^t - p^t} + \frac{\eta_t^2}{2}\InNorms{g^t}^2 \\
        &= M_f(x^t) - \eta_t \InAngles{g^t, x^t - p^t} + \frac{\eta_t^2}{2}\InNorms{g^t}^2,
    \end{align*}
    where the last equality holds since $p^t = \prox_f(x^t)$. Rearranging terms and dividing both sides by $\eta_t$, we get
    \begin{align}
        \InAngles{g^t, x^t - p^t} \le \frac{M_f(x^t) - M_f(x^{t+1})}{\eta_t} + \frac{\eta_t}{2}\InNorms{g^t}^2. \label{eq:GD one step}
    \end{align}
    This is the standard one-step analysis of \hyperref[GD]{Online Gradient Descent} (see \emph{e.g.}, the proof of \citealp[Theorem 3.2]{bubeck2015convex} and \citep[Theorem 2.13]{orabona2019modern}), with the squared distance to a fixed comparator replaced by the Moreau envelope. Summing \eqref{eq:GD one step} over $t \in [T]$ proves the first part of the theorem. 

    When $\+X$ is compact with diameter $D$, let $e_t := M_f(x^t) - \min_{x \in \+X} f(x)$ for $t \in [T+1]$. By item 2 of \Cref{lemma:moreau envelope}, we have $0 \le e_t \le \frac{D^2}{2}$ for all $t \in [T+1]$. Since $\{\eta_t\}$ is non-increasing, we have
    \begin{align*}
        \sum_{t=1}^T \frac{M_f(x^t) - M_f(x^{t+1})}{\eta_t} = \sum_{t=1}^T \frac{e_t - e_{t+1}}{\eta_t} &= \frac{e_1}{\eta_1} + \sum_{t=2}^{T} e_t \InParentheses{\frac{1}{\eta_t} - \frac{1}{\eta_{t-1}}} - \frac{e_{T+1}}{\eta_T} \\
        &\le \frac{D^2}{2}\InParentheses{\frac{1}{\eta_1} + \sum_{t=2}^{T}\InParentheses{\frac{1}{\eta_t} - \frac{1}{\eta_{t-1}}}} = \frac{D^2}{2\eta_T}.
    \end{align*}
    Combining the above inequality with the first part of the theorem completes the proof. Since the choice of $f$ and $\prox_f$ is arbitrary, the bound holds simultaneously for all $f \in \+F_{\lsc}(\+X)$ and all proximal operators $\prox_f$.
\end{proof}

When $f$ is additionally weakly convex, the proof above gives a slightly stronger bound with an additional negative term.
\begin{remark}[Weakly Convex Functions]\label{remark:weakly convex}
    If $f$ is $\rho$-weakly convex with $\rho < 1$, then the objective $x' \mapsto f(x') + \frac{1}{2}\InNorms{x' - x^{t+1}}^2$ is $(1-\rho)$-strongly convex and is minimized over $\+X$ at $p^{t+1} = \prox_f(x^{t+1})$. Thus the first inequality in the proof of \Cref{theorem:GD proximal regret} improves to $M_f(x^{t+1}) \le f(p^t) + \frac{1}{2}\InNorms{x^{t+1} - p^t}^2 - \frac{1-\rho}{2}\InNorms{p^t - p^{t+1}}^2$. Consequently, all the bounds in \Cref{theorem:GD proximal regret} hold with an additional negative term $-\sum_{t=1}^{T} \frac{1-\rho}{2\eta_t}\InNorms{p^t - p^{t+1}}^2$ on the right-hand side, where $p^{T+1} = \prox_f(x^{T+1})$.
\end{remark}

\subsection{Characterization of Proximal Operators}
\label{sec:characterization}
Since \Cref{theorem:GD proximal regret} holds for all proper lower semi-continuous functions, it is natural to ask which strategy modifications arise as their proximal operators. In this section, we present the characterization of the set $\Phiprox$: over a compact convex set, a mapping is a proximal operator if and only if it is \emph{cyclically monotone}, or equivalently, a selection of the subdifferential of a convex function. This characterization follows from the characterization of proximal operators of possibly nonconvex functions by~\citet{gribonval2020characterization}, together with the classical characterization of cyclically monotone operators by~\citet{rockafellar1966characterization}.\footnote{\citet[Theorem 1]{gribonval2020characterization} consider minimization over $\mathbb{R}^d$ and $f$ that needs not be lower semi-continuous. Their result applies to our constrained setting by considering $f + \mathbb{I}_{\+X}$.} 

\begin{definition}[Cyclically Monotone Mappings]
    A mapping $\phi: \+X \rightarrow \-R^d$ is \emph{cyclically monotone} if for any integer $m \ge 1$ and any points $x_0, x_1, \ldots, x_m \in \+X$, letting $x_{m+1} := x_0$, we have
    \begin{align*}
        \sum_{k=0}^{m} \InAngles{\phi(x_k), x_{k+1} - x_k} \le 0.
    \end{align*}
\end{definition}
Taking $m = 1$ in the definition, we see that every cyclically monotone mapping $\phi$ is also \emph{monotone}, that is, $\InAngles{\phi(x) - \phi(y), x - y} \ge 0$ holds for all $x, y \in \+X$.

\begin{proposition}[Characterization of Proximal Operators, adapted from~\citep{gribonval2020characterization}]\label{prop:characterization}
    Let $\+X \subseteq \-R^d$ be a nonempty compact convex set. For any mapping $\phi: \+X \rightarrow \+X$, the following are equivalent:
    \begin{enumerate}
        \item $\phi$ is a proximal operator associated with some $f \in \+F_{\lsc}(\+X)$, i.e., $\phi(x) \in \argmin_{x' \in \+X}\{f(x') + \frac{1}{2}\InNorms{x' - x}^2\}$ for all $x \in \+X$;
        \item there exists a proper, lower semi-continuous, convex function $\psi: \-R^d \rightarrow \bR$ such that $\phi(x) \in \partial \psi(x)$ for all $x \in \+X$, i.e., $\psi(y) \ge \psi(x) + \InAngles{\phi(x), y - x}$ for all $x \in \+X$ and $y \in \-R^d$;
        \item $\phi$ is cyclically monotone.
    \end{enumerate}
\end{proposition}
The characterization in \Cref{prop:characterization} is useful when we have an explicit description of $\phi$ but are unsure if it can be represented as a proximal operator $\prox_f$ for some function $f$, since it is often easier to check that $\phi$ is a selection of the subdifferential of a convex function $\psi$. In the following section, we use \Cref{prop:characterization} to characterize symmetric linear modifications as proximal operators.

\subsection{Symmetric Linear Swap Regret}
\label{sec:symmetric linear swap regret}
In this section, we show that upper bounds on proximal regret imply upper bounds on \emph{symmetric linear swap regret} and that \hyperref[GD]{GD} has $O(\sqrt{T})$ symmetric linear swap regret. This result is non-trivial without using the proximal operators as a proxy for symmetric linear transformations.

We recall that linear swap regret~\citep{daskalakis2025efficient} is $\Phi$-regret where $\Phi$ contains all the affine \emph{endomorphisms} over $\+X$: an affine strategy modification $\phi: x \rightarrow Ax + b$ is an \emph{endomorphism} if $\phi$ maps from $\+X$ to $\+X$ itself. Symmetric linear swap regret is linear swap regret with the additional assumption that $A$ is symmetric. We show that symmetric linear swap regret can be instantiated as proximal regret after applying a transformation that preserves the regret, which by \Cref{theorem:GD proximal regret} implies that \hyperref[GD]{GD} has $O(\sqrt{T})$ symmetric linear swap regret. This result is presented as \Cref{thm:linear swap}. We remark that although the work~\citep{daskalakis2025efficient} has presented an efficient learning algorithm that minimizes the general linear swap regret, it relies heavily on the ellipsoid method and is considerably more complicated than \hyperref[GD]{GD}. We find 
the fact that the simple \hyperref[GD]{GD} algorithm minimizes the symmetric linear swap regret surprising, as it simultaneously minimizes other proximal regret for all proper lower semi-continuous functions. Since \hyperref[GD]{GD} is widely applied in practice, we believe that providing stronger guarantees for this classic algorithm is an important contribution.

\paragraph{Characterization of Symmetric Affine Endomorphisms using Proximal Operators} Let $\+X \subset \-R^d$ be a compact convex set. We consider symmetric affine endomorphisms $\phi: x \rightarrow Ax + b$ such that $A \in \-R^{d \times d}$ is symmetric. In the following proposition, we show that $\phi$ can be expressed as a proximal operator $\prox_f$ whenever $A$ is positive semidefinite (PSD), and we give an explicit quadratic function $f$ when $A$ is positive definite (PD).
\begin{proposition}\label{prop:affine transformation}
    For any symmetric affine endomorphism $\phi(x) = Ax +b$ over $\+X$ such that $A$ is PSD, there exists $f \in \+F_{\lsc}(\+X)$ such that $\phi = \prox_f$. Moreover, if $A$ is PD, then $\phi$ is the unique proximal operator associated with the quadratic function $f(x) = \frac{1}{2} x^\top (A^{-1}-I)x - (A^{-1}b)^\top x$. In this case, $f$ is convex if the largest eigenvalue of $A$ satisfies $\sigma_{\max}(A) \le 1$, and $f$ is $(\frac{1}{\sigma_{\max}(A)} - 1)$-strongly convex if $\sigma_{\max}(A) < 1$.
\end{proposition}
    \begin{proof}
        If $A$ is PSD, then $\phi = \nabla \psi$ for the convex function $\psi(x) = \frac{1}{2}x^\top A x + b^\top x$, and thus $\phi(x) \in \partial \psi(x)$ for all $x \in \+X$. By \Cref{prop:characterization}, there exists $f \in \+F_{\lsc}(\+X)$ such that $\phi = \prox_f$.

        Now suppose $A$ is PD and $f$ is the quadratic function defined above. Fix any $x \in \+X$ and recall that
        \begin{align*}
            \prox_f(x) \in \argmin_{x' \in \+X} F_x(x'): = f(x') + \frac{1}{2}\InNorms{x' - x}^2.
        \end{align*}
        Since $\nabla^2 F_x = A^{-1}$ is PD, $F_x$ is strongly convex. Let us denote $x^* = \phi(x) = Ax + b \in \+X$. Then
        \begin{align*}
            \nabla F_x(x^*) &= (A^{-1} - I) x^* - A^{-1} b + x^* - x \tag{since $A$ is symmetric} \\
            & = A^{-1} (x^* - b)  - x\\
            & = x - x = 0.
        \end{align*}
        Since $F_x(x')$ is a strictly convex function, we know $x^* = Ax + b \in \+X$ is its unique minimizer. Thus $\prox_f(x) = Ax+b = \phi(x)$. Finally, the eigenvalues of $\nabla^2 f = A^{-1} - I$ are $\frac{1}{\sigma} - 1$, where $\sigma > 0$ ranges over the eigenvalues of $A$. They are all nonnegative if $\sigma_{\max}(A) \le 1$, and all at least $\frac{1}{\sigma_{\max}(A)} - 1 > 0$ if $\sigma_{\max}(A) < 1$.
    \end{proof}
\paragraph{Proximal Regret Bounds Symmetric Linear Swap Regret} Now we proceed to show that \hyperref[GD]{GD} efficiently minimizes symmetric linear swap regret. We cannot directly apply \Cref{prop:affine transformation} to any affine transformation $\phi: Ax + b$ since the matrix $A$ may not be PSD. The main idea is that we can consider an auxiliary symmetric affine endomorphism by interpolating $\phi$ with the identity mapping $\mathrm{Id}$: $\phi_\alpha := (1-\alpha) \mathrm{Id} + \alpha \phi$ such that $\phi_\alpha(x) = A_\alpha x+ \alpha b$ where $A_\alpha = (1-\alpha)I + \alpha A$. On one hand, if we choose $\alpha = \frac{1}{1+\InNorms{A}_2}$, then $A_\alpha$ is PSD. On the other hand, since $x - \phi(x) = \frac{1}{\alpha}(x - \phi_\alpha(x))$, the regret $\reg_\phi(T)$ is bounded by $\frac{1}{\alpha}$ times the linearized regret with respect to $\phi_\alpha$ for any $T \ge 1$:
\begin{align*}
    \reg_{\phi}(T) \le \sum_{t=1}^T \InAngles{\nabla \ell^t(x^t) ,x^t - \phi(x^t)} = \frac{1}{\alpha} \sum_{t=1}^T \InAngles{\nabla \ell^t(x^t) ,x^t - \phi_\alpha(x^t)}.
    \end{align*}
Then it suffices to show that \hyperref[GD]{GD} minimizes the linearized regret with respect to $\phi_\alpha$, which holds by \Cref{theorem:GD proximal regret}.

Formally, we have the following theorem. For simplicity of the analysis, we assume the set $\+X$ lies in the $\ell_2$-ball $B_d(0, D)$. This assumption is without loss of generality, as any bounded convex set can be transposed to such a position~\citep[Lemma A.1]{daskalakis2025efficient}. Moreover, in this case, any endomorphism $\phi:Ax +b$ has bounded $\InNorms{A}_2$~\citep[Lemma B.2]{daskalakis2025efficient}.

\begin{theorem}[Gradient Descent Minimizes Symmetric Linear Swap Regret]\label{thm:linear swap}
Let $\+X \subseteq B_d(0,D)$ be a compact convex set and $\{\ell^t: \+X \rightarrow \-R\}$ be a sequence of convex and $G$-Lipschitz loss functions. Let $\{x^t\}$ be the iterates of \hyperref[GD]{GD} with step size $\eta_t = \frac{1}{\sqrt{t}}$.   Then for any symmetric matrix $A$ and vector $b$ such that $\phi(x) = Ax + b$ is an endomorphism over $\+X$, we have
\[
    \reg_\phi(T)  \le  (1+\InNorms{A}_2)(2D^2 + G^2) \cdot \sqrt{T}.
\]
\end{theorem}
\begin{remark}
    When $\+X = \Delta^d$ is the $d$-dimensional simplex, every symmetric affine endomorphism over $\Delta^d$ admits a representation $\phi(x) = Ax + b$ with a symmetric matrix $A$ such that $\InNorms{A}_2 \le 1$. In this case, the set $\Delta^d \subseteq B_d(0, 1)$ and we have $D = 1$. Therefore, the symmetric linear swap regret of \hyperref[GD]{GD} with step size $\eta_t = \frac{1}{\sqrt{t}}$ over the simplex $\Delta^d$ against convex and $G$-Lipschitz losses is bounded by $2(2+G^2)\sqrt{T}$ for any $T \ge 1$. This can be improved to $O(G\sqrt{T})$ if we know $G$ and use the fixed step size $\eta = \frac{1}{G\sqrt{T}}$.
\end{remark}
\begin{proof}[Proof of \Cref{thm:linear swap}]
    Given any symmetric affine endomorphism $\phi(x) = Ax +b$, we consider a new symmetric affine endomorphism $\phi_\alpha := (1-\alpha) \mathrm{Id} + \alpha \phi$ such that $\phi_\alpha(x) = A_\alpha x+ \alpha b$ where $A_\alpha = (1-\alpha)I + \alpha A$. Since $\+X$ is convex, $\phi_\alpha$ is also an endomorphism over $\+X$. We note that for any $T \ge 1$,
    \begin{align*}
    \reg_{\phi}(T) \le \sum_{t=1}^T \InAngles{\nabla \ell^t(x^t) ,x^t - \phi(x^t)} = \frac{1}{\alpha} \sum_{t=1}^T \InAngles{\nabla \ell^t(x^t) ,x^t - \phi_\alpha(x^t)}.
    \end{align*}
    Choosing $\alpha = \frac{1}{1+\InNorms{A}_2}$, we have the eigenvalues of $A_\alpha$ are at least $1-\alpha - \alpha \InNorms{A}_2 = 0$. Thus $A_\alpha$ is PSD, and by \Cref{prop:affine transformation}, $\phi_\alpha = \prox_f$ for some $f \in \+F_{\lsc}(\+X)$. Since $\+X \subseteq B_d(0,D)$ has diameter at most $2D$, by \Cref{theorem:GD proximal regret} and the calculation in \Cref{remark:GD}, \hyperref[GD]{GD} with step size $\eta_t = \frac{1}{\sqrt{t}}$ satisfies for any $T \ge 1$,
    \begin{align*}
        \sum_{t=1}^T \InAngles{\nabla \ell^t(x^t) ,x^t - \phi_\alpha(x^t)} \le (2D^2+G^2)\sqrt{T}.
    \end{align*}
   Combining the above inequalities, we have $\reg_\phi(T) \le (1+\InNorms{A}_2)(2D^2 + G^2) \cdot \sqrt{T}$. This completes the proof.
\end{proof}

\section{Improved Regret and Faster Convergence in Games}
\label{sec:OG game setting}
Any online algorithm suffers $\Omega(\sqrt{T})$ proximal regret in the worst case since proximal regret covers the external regret as a special case. This lower bound, however, holds only in the \emph{adversarial} setting.
In this section, we study improved regret and convergence in the \emph{game} setting, where players interact with each other using the same algorithm. 
We remark that improved regret guarantees have been shown for the external regret~\citep{syrgkanis2015fast,chen2020hedging,daskalakis2021near-optimal, farina2022near} in general convex games, giving $O(\log T)$ individual external regret and $O(1)$ social external regret bound. However, results beyond the external regret are sparse. Our result fits in the line of literature as an extension from external to proximal regret for general convex games. Specifically, for general convex games, we obtain $\reg_{\+F_{\convex}}^T=O(T^{1/4})$ individual proximal regret bound for convex functions $\+F_{\convex}$, and $O(1)$ social proximal regret $\reg_{\+F_{\strc}}^T$ for strongly convex functions $\+F_{\strc}$, which is stronger than external regret. Both results are new and complement existing results.

\paragraph{Optimistic Gradient} We study the Optimistic Gradient Descent \eqref{OG} algorithm~\citep{rakhlin2013optimization}, an optimistic variant of \hyperref[GD]{GD} that has been shown to have improved $O(T^{1/4})$ individual \emph{external} regret guarantee in the game setting~\citep{syrgkanis2015fast}. Given non-increasing step sizes $\{\eta_t >0\}$, the \ref{OG} algorithm initializes $w^0 \in \X$ arbitrarily and $g^0 = 0$. In each step $t \ge 1$, the algorithm plays $x^t$, receives gradient feedback $g^t := \nabla \ell^t(x^t)$, and updates $w^t$, as follows:
\begin{equation}
\label{OG}
\tag{OG}
    \begin{aligned}
        x^t = \Pi_\X \InBrackets{ w^{t-1} - \eta_t g^{t-1}},  \quad w^t = \Pi_\X \InBrackets{ w^{t-1} - \eta_t g^t}. 
    \end{aligned}
\end{equation}
\paragraph{Improved Individual Regret } We first give an improved instance-dependent regret bound of \ref{OG}. Specifically, we show that for any convex function $f \in \+F_{\convex}$, the adversarial proximal regret of \ref{OG} is $\reg^T_{\+F_{\convex}} = O(\sqrt{P^T})$, where $P^T:=\sum_{t=1}^T \InNorms{g^t- g^{t-1}}^2$  is the total sum of gradient variation. 

\begin{theorem}[Adversarial Regret Bound for \ref{OG}]
\label{thm:adversaril regret of OG}
    Let $\+X \subseteq \-R^d$ be a closed convex set and $\{\ell^t: \+X \rightarrow \-R\}$ be a sequence of convex loss functions. For any convex function $f \in \+F_{\convex}(\+X)$, denote $p^t = \prox_f(x^t)$, \ref{OG} guarantees for all $T \ge 1$,
    \begin{align*}
        \sum_{t=1}^T \InAngles{\nabla \ell^t(x^t), x^t - \prox_f(x^t)} \le \frac{D^2 + 2B_f - \InNorms{w^T-p^T}^2}{2\eta_T} + \sum_{t=1}^T\eta_t\InNorms{g^t - g^{t-1}}^2 - \sum_{t=1}^{T} \frac{1}{2\eta_t} \InNorms{x^t - w^{t}}^2,
    \end{align*}
    where $D = \max_{0\le t\le T-1} \InNorms{w^t-p^{t+1}}$ and $B_f = \max_{t \in [T]}f(p^t) - \min_{t \in [T]} f(p^t)$. If the step size is constant $\eta_t = \eta$, then the above bound holds with $D^2 + 2B_f - \InNorms{w^T-p^T}^2$ replaced by $2(M_f(w^0) - M_f(w^T))$. Moreover, if $\+X$ is compact with diameter $\Dx$, then $D^2 + 2B_f \le 2\Dx^2$ and $M_f(w^0) - M_f(w^T) \le \frac{\Dx^2}{2}$ by \Cref{lemma:moreau envelope}, so the bounds are independent of $f$.
\end{theorem}
\begin{remark}
   Compared to the regret of \hyperref[GD]{GD} which has dependence on $\sum_{t=1}^T\eta_t \InNorms{g^t}^2$, the main improvement of \ref{OG} is that the dependence becomes $P^T = \sum_{t=1}^T\InNorms{g^t - g^{t-1}}^2$. We note that $P^T$ is at most $O(T)$ but may be much smaller when the environment is stable.
\end{remark}
\begin{remark}\label{rem:OG convex}
    Unlike \Cref{theorem:GD proximal regret}, \Cref{thm:adversaril regret of OG} requires $f$ to be convex. In the analysis of \ref{OG}, the deviation $p^{t+1} = \prox_f(x^{t+1})$ is determined by the played strategy $x^{t+1}$, while the distance in the potential is measured from the auxiliary iterate $w^t$. When $f$ is convex, the first-order optimality condition of the proximal operator together with the convexity of $f$ bounds the resulting error term by $\InNorms{w^t - x^{t+1}}^2$, which is canceled by a negative term in the analysis of \ref{OG}. For nonconvex $f$, the proximal operator may be discontinuous, so $p^{t+1}$ may differ significantly from $\prox_f(w^t)$ even though $x^{t+1}$ is close to $w^t$, and this argument breaks down. We leave it as an interesting open question whether \ref{OG} minimizes proximal regret with respect to all proper lower semi-continuous functions.
\end{remark}
\begin{remark}\label{rem:RVU}
    We remark that \Cref{thm:adversaril regret of OG} is weaker than the ``Regret bounded by Variation of Utilities'' (RVU)  bound~\citep{syrgkanis2015fast} since the negative terms $-\sum_{t=1}^T \frac{1}{2\eta_t}\InNorms{x^t -w^t}^2$ is not $-\sum_{t=1}^T \frac{1}{2\eta_t}\InNorms{x^t -x^{t-1}}^2$ and does not fully cancel the positive terms $\sum_{t=1}^T\eta_t\InNorms{g^t - g^{t-1}}^2$ even in smooth games. Note that  $\reg_{\+F_{\convex}}^T$ is non-negative. If one was able to prove an RVU bound for $\reg_{\+F_\convex}^T$, then this would lead to constant $O(1)$ individual external regret in general convex games (which is not known to be achievable in the literature).
\end{remark}

In the game setting, when every player employs \ref{OG}, the gradient variation is small since players update their strategies in a stable manner. Using techniques from~\citep{syrgkanis2015fast}, we use this improved gradient variation dependent regret bound to get an improved $O(T^{1/4})$ individual proximal regret and thus faster convergence to approximate proximal correlated equilibrium in games.

\begin{theorem}[Improved Individual Proximal Regret of \ref{OG} in the Game Setting]
\label{thm:game regret of OG}
    In a $G$-Lipschitz $L$-smooth convex game $\+G = \{[n], \{\+X_i\}, \{u_i\}\}$, when all players employ \ref{OG} with fixed step size $\eta > 0$, then for each player $i$, any $f \in \+F_{\convex}(\+X_i)$, and $T \ge 1$, their proximal regret is bounded as
    \begin{align*}
        \sum_{t=1}^T \InAngles{\nabla_{x_i} u_i(x^t), \prox_f(x^t_i)-x^t_i }\le \frac{D_{\+X_i}^2}{2\eta} +  \eta G^2 + 9n L^2 G^2 \eta^3 T,
    \end{align*}
    where $D_{\+X_i}$ is the diameter of $\+X_i$.
    Choosing $\eta = T^{-\frac{1}{4}}$, the above regret is upper bounded by $(\frac{D_{\+X_i}^2}{2} + G^2 + 9n L^2 G^2)T^{\frac{1}{4}}$. Furthermore, let $\Phi = \{\prox_f: f \in \+F_{\convex}\}$, the empirical distribution of strategy profiles converges to an $\varepsilon$-approximate $\Phi$-equilibrium in 
    $O(1/\varepsilon^{\frac{4}{3}})$ iterations.
\end{theorem}

\paragraph{Improved Social Regret} Let $\+F_{\alpha,\strc}$ be the set of $\alpha$-strongly convex functions for $\alpha > 0$. In the following, we show that when each player employs \ref{OG} in a smooth convex game, the social regret $\sum_{i \in [n]} \reg_{i, \+F_{\alpha, \strc}}^T = O(1)$ (omitting dependence on other parameters for simplicity), where $\reg_{i, \+F_{\alpha, \strc}}^T$ is the individual proximal regret with respect to $\+F_{\alpha, \strc}(\+X_i)$ of player $i$.

We first remark that our result improves the existing $O(1)$ social external regret bound~\citep{syrgkanis2015fast} even though $\+F_{\alpha,\strc}$ no longer contains the indicator functions (that we used earlier to recover external regret).
To see this, note that every constant strategy modification $\phi(x) = b$ is a symmetric affine endomorphism. Consider a new strategy modification $\phi_\alpha := (1-\alpha)I+\alpha b$. By discussion in \Cref{sec:symmetric linear swap regret} and \Cref{prop:affine transformation}, we know that $\phi_\alpha=\prox_f$ for a $\alpha$-strongly convex quadratic function $f$, and the external regret against $b$ is upper bounded by $\frac{1}{\alpha}\reg_f^T$. Thus the external regret $\reg_{\ext}^T$ is at most $\frac{1}{\alpha} \cdot \reg_{\+F_{\alpha, \strc}}^T$. 

Moreover, while \citep{syrgkanis2015fast} achieves $O(1)$ social external using the RVU bound (\Cref{rem:RVU}), as mentioned, our bound in \Cref{thm:adversaril regret of OG} is weaker than RVU.
Nevertheless, by improving the analysis in \Cref{thm:adversaril regret of OG} and obtaining an improved regret bound with an additional negative term due to strong convexity, we still manage to show an RVU proximal regret bound for strongly convex $f$ and consequently $O(1)$ social proximal regret in smooth convex games.


\begin{theorem}[$O(1)$ social proximal regret with respect to strongly convex functions]\label{thm:o(1)social regret}
    In a $G$-Lipschitz and $L$-smooth convex game $\+G = \{[n], \{\+X_i\}, \{u_i\}\}$, when each player employs \ref{OG} with step size $\eta \le \sqrt{\frac{\min\{\alpha,1\}}{8n L^2}}$, we have for any $\alpha$-strongly convex functions $\{f_i\}$ and $T \ge 1$
    \begin{align*}
        \sum_{i=1}^n \reg_{i,f_i}^T \le \frac{\sum_{i=1}^n D_{\+X_i}^2}{2\eta} + n\eta G^2,
    \end{align*}
    where $D_{\+X_i}$ is the diameter of $\+X_i$ for each $i \in [n]$.
\end{theorem}

\subsection{Proof of \Cref{thm:adversaril regret of OG} }
\begin{proof}
Fix a function $f \in \+F_{\convex}(\+X)$. We define $p^t = \prox_f(x^t)$. Following standard analysis of \ref{OG}~\citep[Lemma 1]{rakhlin2013optimization}, we have
\begin{align}
    & \sum_{t=1}^T \ell^t(x^t) - \ell^t(p^t) \le \sum_{t=1}^T \InAngles{\nabla \ell^t(x^t), x^t - p^t}  \nonumber \\
    & \le \sum_{t=1}^T \frac{1}{2\eta_t} \InParentheses{ \InNorms{w^{t-1} - p^t}^2 - \InNorms{w^t - p^t}^2 } + \eta_t \InNorms{g^t - g^{t-1}}^2 - \frac{1}{2\eta_t}\InParentheses{  \InNorms{x^t - w^t}^2 + \InNorms{x^t - w^{t-1}}^2}  \nonumber \\
    & = \frac{\InNorms{w^0 - p^1}^2}{2\eta_1} + \sum_{t=1}^{T-1} \frac{1}{2\eta_t}\InParentheses{ \InNorms{w^{t} - p^{t+1}}^2 - \InNorms{w^{t} - p^{t}}^2} + \sum_{t=1}^{T-1} \InParentheses{ \frac{1}{2\eta_{t+1}} - \frac{1}{2\eta_t}} \InNorms{w^t-p^{t+1}}^2\nonumber\\
    & \quad \quad - \frac{1}{2\eta_T} \InNorms{w^T - p^T}^2 + \sum_{t=1}^T \eta_t \InNorms{g^t - g^{t-1}}^2 - \sum_{t=1}^T\frac{1}{2\eta_t}\InParentheses{\InNorms{x^t - w^t}^2 + \InNorms{x^t - w^{t-1}}^2} \nonumber \\
    &\le  \frac{D^2 - \InNorms{w^T-p^T}^2}{2\eta_T} + \sum_{t=1}^{T-1} \frac{1}{2\eta_t}\InParentheses{ \InNorms{w^{t} - p^{t+1}}^2 - \InNorms{w^{t} - p^{t}}^2}+ \sum_{t=1}^T \eta_t \InNorms{g^t - g^{t-1}}^2 \nonumber  \\
    &\quad \quad -\sum_{t=1}^T \frac{1}{2\eta_t}\InParentheses{  \InNorms{x^t - w^t}^2 + \InNorms{x^t - w^{t-1}}^2}.\nonumber
\end{align}
In the last inequality, we use the fact that $\{\eta_t\}$ is non-increasing and $D:= \max_{0\le t\le T-1} \InNorms{w^t - p^{t+1}}$. If $\eta_t = \eta$ is constant, then the above inequality holds with $D = \InNorms{w^0 - p^1}$.

Now we use the convexity of $f$ to
upper bound the term $\InNorms{w^{t} - p^{t+1}}^2 - \InNorms{w^{t} - p^{t}}^2$. Since $p^{t+1} = \prox_f(x^{t+1})$ and $f$ is convex, by \Cref{fact:proximal operator} we know there exists $v \in \partial f(p^{t+1})$ such that
\begin{align*}
    \InAngles{x' - p^{t+1},x^{t+1} - v - p^{t+1} } \le 0, \quad \forall x' \in \+X.
\end{align*}
We now proceed as:
\begin{align*}
    & \InNorms{w^{t} - p^{t+1}}^2 - \InNorms{w^{t} - p^{t}}^2 \\
     &= 2\InAngles{p^{t} - p^{t+1},  w^{t} - p^{t+1}} - \InNorms{p^t -  p^{t+1}}^2 \\
     &= 2\InAngles{p^{t} - p^{t+1}, v} + 2\InAngles{p^{t} - p^{t+1},  w^{t} - v - p^{t+1}} - \InNorms{p^{t} - p^{t+1}}^2 \\
     &= 2\InAngles{p^{t} - p^{t+1}, v} + 2\InAngles{p^{t} - p^{t+1},  x^{t+1} - v - p^{t+1}} + 2\InAngles{p^{t} - p^{t+1}, w^{t} -  x^{t+1}}- \InNorms{p^{t} - p^{t+1}}^2 \\
     &\le 2\InAngles{p^{t} - p^{t+1}, v}  + 2\InAngles{p^{t} - p^{t+1}, w^{t} -  x^{t+1}}- \InNorms{p^{t} - p^{t+1}}^2 \\
     & \le 2(f(p^t) - f(p^{t+1})) + \InNorms{w^t - x^{t+1}}^2,
\end{align*}
where in the second last-inequality we use $\InAngles{p^{t} - p^{t+1},  x^{t+1} - v - p^{t+1}} \le 0$; in the last inequality, we use convexity of $f$ and the basic inequality $2\InAngles{a,b} - \InNorms{b}^2 \le \InNorms{a}^2$.

Now we combine the above two inequalities and get
\begin{align*}
&\sum_{t=1}^T \InAngles{\nabla \ell^t(x^t), x^t - p^t} \\
&\le \frac{D^2 - \InNorms{w^T-p^T}^2}{2\eta_T} + \sum_{t=1}^{T-1}\frac{1}{\eta_t}\InParentheses{f(p^t) - f(p^{t+1})} + \sum_{t=1}^T \eta_t \InNorms{g^t - g^{t-1}}^2  -\sum_{t=1}^T \frac{1}{2\eta_t} \InNorms{x^t - w^t}^2,
\end{align*}
where we use that $\{\eta_t\}$ is non-increasing, so that each term $\frac{1}{2\eta_t}\InNorms{w^t - x^{t+1}}^2$ is canceled by the negative term $-\frac{1}{2\eta_{t+1}}\InNorms{x^{t+1} - w^t}^2$.
Since $\{\eta_t\}$ is non-increasing and $f(p^t) - \min_{s \in [T]} f(p^s) \le B_f$ for all $t \in [T]$ where $B_f:= \max_{t \in [T]} f(p^t) - \min_{t \in [T]} f(p^t)$, we can telescope the second term as in the proof of \Cref{theorem:GD proximal regret} and get
\begin{align*}
&\sum_{t=1}^T \InAngles{\nabla \ell^t(x^t), x^t - p^t} \le \frac{D^2 + 2B_f - \InNorms{w^T-p^T}^2}{2\eta_T} + \sum_{t=1}^T \eta_t \InNorms{g^t - g^{t-1}}^2 -\sum_{t=1}^T \frac{1}{2\eta_t} \InNorms{x^t - w^t}^2.
\end{align*}
When the step size is constant $\eta_t = \eta$, the second term telescopes directly, and the above inequality holds with $D^2 = \InNorms{w^0 - p^1}^2$ and $B_f$ replaced by $f(p^1) - f(p^T)$. Since $g^0 = 0$, we have $x^1 = w^0$ and thus $f(p^1) + \frac{1}{2}\InNorms{w^0 - p^1}^2 = M_f(w^0)$. Moreover, $f(p^T) + \frac{1}{2}\InNorms{w^T - p^T}^2 \ge M_f(w^T)$ by \Cref{lemma:moreau envelope}. Therefore,
\begin{align*}
    \InNorms{w^0 - p^1}^2 + 2\InParentheses{f(p^1) - f(p^T)} - \InNorms{w^T-p^T}^2 \le 2\InParentheses{M_f(w^0) - M_f(w^T)}.
\end{align*}
Finally, if $\+X$ is compact with diameter $\Dx$, then $D \le \Dx$ since $w^t, p^{t+1} \in \+X$. Moreover, by item 2 of \Cref{lemma:moreau envelope}, we have $B_f \le \max_{t \in [T]} f(p^t) - \min_{x \in \+X} f(x) \le \frac{\Dx^2}{2}$ and $M_f(w^0) - M_f(w^T) \le \frac{\Dx^2}{2}$. This completes the proof.
\end{proof}
\subsection{Proof of \Cref{thm:game regret of OG}}
\begin{proof}
Let us fix any player $i \in [n]$ in the smooth game. In every step $t$, player $i$'s loss function $\ell_i^t: \X_i \rightarrow \R$  is $\InAngles{- \nabla_{x_i}u_i(x^t), \cdot }$ determined by their utility function $u_i$ and all players' actions $x^t$. Therefore,  their gradient feedback is $g^t = -\nabla_{x_i} u_i(x^t)$. 
For all $t \ge 2$, we have
\begin{align*}
    \InNorms{g^t - g^{t-1}}^2 &= \InNorms{\nabla_{x_i} u_i(x^t) - \nabla_{x_i} u_i(x^{t-1})}^2 \\
    & \le L^2\InNorms{x^t - x^{t-1}}^2 \\
    & = L^2\sum_{j=1}^n \InNorms{x^t_j - x^{t-1}_j}^2 \\
    & \le 3 L^2\sum_{j=1}^n \InParentheses{ \InNorms{x^t_j - w^{t-1}_j}^2 + \InNorms{w^{t-1}_j - w^{t-2}_j}^2 + \InNorms{w^{t-2}_j - x^{t-1}_j}^2}\\
    & \le 9n L^2 \eta^2 G^2,
\end{align*}
where we use $L$-smoothness of the utility function $u_i$ in the first inequality, and the update rule of \ref{OG} as well as the fact that gradients are bounded by $G$ in the last inequality (each of the three terms in the parentheses is at most $\eta^2 G^2$). For $t = 1$, we have $\InNorms{g^1 - g^0}^2 = \InNorms{g^1}^2 \le G^2$.

Applying the above inequalities to the regret bound obtained in \Cref{thm:adversaril regret of OG} with constant step size $\eta$, and using $M_f(w^0_i) - M_f(w^T_i) \le \frac{D_{\+X_i}^2}{2}$, the individual proximal regret of player $i$ for a convex function $f$ is upper bounded by
\begin{align*}
     \sum_{t=1}^T \InAngles{\nabla_{x_i} u_i(x^t), \prox_f(x^t_i)-x^t_i }\le \frac{D_{\+X_i}^2}{2\eta} +  \eta G^2 + 9n L^2 G^2 \eta^3 T.
\end{align*}
Choosing $\eta = T^{-\frac{1}{4}}$, the left hand side is bounded by $(\frac{D_{\+X_i}^2}{2} + G^2 + 9nL^2G^2) T^{\frac{1}{4}}$. Furthermore, let $\Phi = \{\prox_f: f \in \+F_{\convex}\}$, the empirical distribution of played strategy profiles converges to an $\varepsilon$-approximate $\Phi$-equilibrium in $O(1/\varepsilon^{\frac{4}{3}})$ iterations.
\end{proof}

\subsection{Proof of \Cref{thm:o(1)social regret}}
We first present the adversarial proximal regret w.r.t. strongly convex functions of \ref{OG}.
\begin{theorem}
    Let $\{\ell^t: \+X \rightarrow \-R\}$ be a sequence of convex loss functions. Let $\{x^t\}$ be the sequence produced by \ref{OG} with a fixed learning rate $\eta$. Denote $p^t = \prox_f(x^t)$. Then the proximal regret of \ref{OG} against any $\alpha$-strongly convex function $f$ is
    \begin{align*}
        \reg_{f}^T &\le \sum_{t=1}^T \InAngles{\nabla \ell^t(x^t), x^t - \prox_f(x^t)} \\
        &\le \frac{M_f(w^0) - M_f(w^T)}{\eta} + \sum_{t=1}^T \eta \InNorms{g^t - g^{t-1}}^2 - \sum_{t=1}^T\frac{1}{2\eta}\InNorms{x^t - w^t}^2 - \sum_{t=1}^{T-1} \frac{\min\{\alpha,1\}}{4\eta} \InNorms{x^{t+1} - w^{t}}^2 \\
        & \le \frac{M_f(w^0) - M_f(w^T)}{\eta} + \sum_{t=1}^T \eta \InNorms{g^t - g^{t-1}}^2 - \sum_{t=1}^{T-1}\frac{\min\{\alpha,1\}}{8\eta}\InNorms{x^{t+1} - x^t}^2.
    \end{align*}
    If $\+X$ is compact with diameter $\Dx$, then $M_f(w^0) - M_f(w^T) \le \frac{\Dx^2}{2}$.
\end{theorem}
\begin{proof}  This is a direct adaptation of the proof for the convex $f$ cases in \Cref{thm:adversaril regret of OG}. We get the additional $-\sum_{t=1}^{T-1} \frac{\min\{\alpha,1\}}{4\eta}\InNorms{x^{t+1} - w^t}^2$ term in the step where we bound $\sum_{t=1}^{T-1}\InParentheses{\InNorms{w^{t} - p^{t+1}}^2 - \InNorms{w^{t} - p^{t}}^2}$ since $f$ is $\alpha$-strongly convex. Specifically, the $\alpha$-strong convexity of $f$ gives $2\InAngles{p^t - p^{t+1}, v} \le 2(f(p^t) - f(p^{t+1})) - \alpha\InNorms{p^t - p^{t+1}}^2$, and thus, using $2\InAngles{a,b} - (1+\alpha)\InNorms{a}^2 \le \frac{1}{1+\alpha}\InNorms{b}^2$, we get
\begin{align*}
    \InNorms{w^{t} - p^{t+1}}^2 - \InNorms{w^{t} - p^{t}}^2 \le 2(f(p^t) - f(p^{t+1})) + \frac{1}{1+\alpha}\InNorms{w^t - x^{t+1}}^2.
\end{align*}
Compared with the convex case, this leaves the following additional term after the cancellation:
\begin{align*}
    -\frac{1}{2\eta}\InParentheses{1 - \frac{1}{1+\alpha}}\InNorms{x^{t+1} - w^t}^2 \le -\frac{\min\{\alpha,1\}}{4\eta}\InNorms{x^{t+1} - w^t}^2,
\end{align*}
where we use $\frac{\alpha}{1+\alpha} \ge \frac{\min\{\alpha,1\}}{2}$. Finally, the last inequality in the theorem follows from $\min\{\alpha,1\} \le 1$ and
\begin{align*}
    \InNorms{x^{t+1} - x^t}^2 \le 2\InNorms{x^{t+1} - w^t}^2 + 2\InNorms{w^t - x^t}^2.
\end{align*}
\end{proof}

\begin{proof}[Proof of \Cref{thm:o(1)social regret}]
    In an $L$-smooth game, we have 
    \begin{align*}
        \InNorms{\nabla_i u_i(x^{t+1}) - \nabla_i u_i(x^{t})}^2 \le L^2 \InNorms{x^{t+1} - x^t}^2, \forall i \in[n].
    \end{align*}
    Thus the sum of players' proximal regret is 
    \begin{align*}
        \sum_{i=1}^n \reg^T_{i,f_i} &\le \frac{\sum_{i=1}^n D_{\+X_i}^2}{2\eta} + \eta  \sum_{i=1}^n\sum_{t=1}^T \InNorms{g^t_i - g^{t-1}_i}^2 - \frac{\min\{\alpha,1\}}{8\eta} \sum_{i=1}^n\sum_{t=1}^{T-1}  \InNorms{x_i^{t+1} - x_i^t}^2 \\
        &\le \frac{\sum_{i=1}^n D_{\+X_i}^2}{2\eta} + n\eta G^2 +n\eta L^2 \sum_{t=1}^{T-1} \InNorms{x^{t+1} - x^t}^2 - \frac{\min\{\alpha,1\}}{8\eta} \sum_{t=1}^{T-1}  \InNorms{x^{t+1} - x^t}^2 \\
        &\le \frac{\sum_{i=1}^n D_{\+X_i}^2}{2\eta} + n\eta G^2 + \InParentheses{n\eta L^2 - \frac{\min\{\alpha,1\}}{8\eta}} \sum_{t=1}^{T-1}  \InNorms{x^{t+1} - x^t}^2 \\
        &\le \frac{\sum_{i=1}^n D_{\+X_i}^2}{2\eta} + n\eta G^2,
    \end{align*}
    where we use (1) $M_{f_i}(w_i^0) - M_{f_i}(w_i^T) \le \frac{D_{\+X_i}^2}{2}$ for all $i \in [n]$; (2) the gradient norm $\InNorms{\nabla_i u_i(x)} \le G, \forall i \in [n]$; (3) the $L$-smoothness of the utilities; (4) and $\eta \le \sqrt{\frac{\min\{\alpha,1\}}{8nL^2}}$ in the above inequalities.
\end{proof}

\printbibliography

\appendix
\section{Comparison with~\citep{cai2024tractable} and~\citep{ahunbay2024first}}\label{app:comparison}
In this section, we discuss and compare two closely related works. 

\paragraph{Comparison with \citep{cai2024tractable}}
The work of~\citep{cai2024tractable} studies $\Phi$-regret and $\Phi$-equilibrium in non-convex games. 
They find that for smooth non-convex games, the problem of learning an approximate $\Phi$-equilibrium in the local regime can be reduced to no-$\Phi$-regret learning against convex---in fact, linear---losses. 
For online convex optimization, \cite{cai2024tractable} propose two new notions of $\Phi$-regret, the projection-based regret and the interpolation-based regret. 
\begin{itemize}
    \item \textbf{Projection-Based Regret}: $\Phi = \{ \phi_v(\cdot) = \Pi_{\+X}[\cdot-v] : \InNorms{v} \le 1\}$ contains all strategy modifications that add a fixed direction deviation and project when necessary. 
    \item \textbf{Interpolation-Based Regret}: $\Phi = \{\phi_{\alpha,x}(\cdot) = (1-\alpha) (\cdot) + \alpha x: x\in \+X, \alpha \le 1\}$ contains all strategy modifications that interpolate between the current point and a fixed point. 
\end{itemize}
\citep{cai2024tractable} show that \hyperref[GD]{Online Gradient Descent (GD)} efficiently minimizes these two notions of regret and show that projection-based regret is incomparable with the external regret. Our work significantly generalizes \citep{cai2024tractable} by the new notion of proximal regret. As we have discussed in the introduction and \Cref{sec:OG game setting}, the projection-based regret is proximal regret w.r.t linear functions, while the interpolation-based regret is equivalent to external regret and is covered by proximal regret w.r.t strongly convex functions. Moreover, using the reduction in \citep[Lemma 1]{cai2024tractable}, our results also show that \hyperref[GD]{GD} converges to approximate local proximal CE for smooth nonconvex games in the local regime (see \Cref{app:proximal local} for details).

\paragraph{Comparison with \citep{ahunbay2024first}.} 
Building on~\citep{cai2024tractable}, the early version of \citet{ahunbay2024first} proposed a generalized set of strategy modifications based on gradient fields, which includes both projection-based and interpolation-based regret as special cases. 
His work encouraged us to consider a broad family of strategy modifications based on proximal operators---a family much more general than those in~\citep{cai2024tractable}---which we present here. 
While the two approaches are different and largely incomparable, the key distinction is that the early version of \citet{ahunbay2024first} studies different notions of equilibrium and focuses on providing inequality constraints for gradient flow in games, and does not address $\Phi$-regret in the adversarial online learning setting or the approximation of standard $\Phi$-equilibria, which are the main focus of our work. 

More specifically, \citet{ahunbay2024first} proposes two first-order equilibrium notions called \emph{local CCE} and \emph{stationary CCE}. Consider a $G$-Lipschitz and $L$-smooth game $\+G$, a family of Lipschitz continuous gradient fields $\{\nabla f: f \in F\}$ where $F$ is a set of continuous differentiable functions with bounded gradient norm $G_f = \max_{x \in \+X} \InNorms{\nabla f(x)}$ and Lipschitz constant $L_f$. A distribution $\sigma$ over $\+X$ is an $\varepsilon$-\emph{local CCE} if 
\begin{align}\label{eq:Local CCE}
    \sum_{i\in[n]} \-E_{x \sim \sigma} \InBrackets{\InAngles{\Pi_{T_{\+X_i}(x_i)}[\nabla f_i(x)], \nabla_i u_i(x)   }}\le \varepsilon \cdot \poly(G,L, G_f, L_f), \forall f \in F.\footnotemark
\end{align}\footnotetext{$T_\+X(x)$ is the tangent cone of $x \in \+X$}
Let $\Phi_F = \{\phi_f(\cdot)= \Pi_{\+X}[\cdot - \nabla f(\cdot)] : f \in F\}$. The definition of $\varepsilon$-\emph{local CCE} can be seen as the first-order approximation of the $\Phi_F$-equilibria.  A distribution $\sigma$ over $\+X$ is an $\varepsilon$-\emph{stationary CCE} if 
\begin{align}\label{eq:stationary CCE}
    \left | \sum_{i\in[n]} \-E_{x \sim \sigma} \InBrackets{\InAngles{\nabla f_i(x), \Pi_{T_{\+X_i}(x_i)}[\nabla_i u_i(x)]   }} \right|\le \varepsilon \cdot \poly(G,L, G_f, L_f), \forall f \in F.
\end{align}
The local CCE notion is related to \emph{expected variational inequality (EVI)} developed in the subsequent work by~\citet{zhang2025expected}. The stationary CCE is a different notion from both $\Phi$-equilibria and EVI.

Concurrent to our results, a later version of \citet{ahunbay2024first} presents an adversarial proximal regret guarantee for \hyperref[GD]{GD} but requires additional regularity assumptions on the action sets. Building on our results, Ahunbay extended his original analysis and obtained an adversarial proximal regret guarantee for \hyperref[GD]{GD} for general convex sets in the latest version of~\citep{ahunbay2024first}. Moreover, while our analysis naturally generalizes to \hyperref[MD]{Mirror Descent (MD)} and the Bregman setting, his analysis does not fully extend to the Bregman case but requires the distance-generating function $\phi$ to be ``steep''.

\section{Proximal Operator-Based Local Strategy Modifications}\label{app:proximal local}
Following~\citep{cai2024tractable}, we consider \hyperref[GD]{Online Gradient Descent (GD)} in smooth \emph{non-convex} games, which is a smooth game where players' utilities may not be concave.

We consider the set of local strategy modifications based on the proximal operator. Formally, the set $\Phiprox(\delta)$ encompasses all strategy modifications $\prox_f$ that are $\delta$-local.
\begin{align*}
    \Phiprox(\delta) = \{\prox_f: f \in \+F_{\lsc}, \InNorms{\prox_f(x) -x} \le \delta,\forall x \in \+X\}
\end{align*}
Here, $\prox_f$ ranges over all proximal operators associated with $f$. We remark that for any (not necessarily convex) function $f$ that is differentiable on an open set containing $\+X$ with bounded gradients such that $\InNorms{\nabla f(x)} \le G_f, \forall x \in \+X$, any proximal operator $\prox_{a f}$ is $\delta$-local for any scalar $0 < a \le \frac{\delta}{G_f}$. Indeed, $p = \prox_{af}(x)$ minimizes $af(x') + \frac{1}{2}\InNorms{x' - x}^2$ over $\+X$, so the first-order optimality condition at $p$ gives $\InAngles{a\nabla f(p) + p - x, x - p} \ge 0$, which implies $\InNorms{x - p}^2 \le a\InAngles{\nabla f(p), x - p} \le a G_f \InNorms{x - p}$.

We define $\Phiproxeq(\delta) = \Pi_{i=1}^n \Phi^{\+X_i}_{\prox}(\delta)$. By Lemma 1 in \cite{cai2024tractable} and \Cref{theorem:GD proximal regret}, we immediately get the following corollary on learning approximate $\Phiproxeq(\delta)$-equilibrium in the \emph{local regime}, i.e., approximation error $\ge \frac{L\delta^2}{2}$. 
\begin{corollary}
\label{corollary:local proximal equilibrium}
    For any $\delta, \varepsilon > 0$, when each player in a $L$-smooth and $G$-Lipschitz game employ \hyperref[GD]{Online Gradient Descent}, their empirical distribution of played strategy profiles converges to an $(\varepsilon+\frac{L\delta^2}{2})$-approximate $\Phiproxeq(\delta)$-equilibrium in $O(1/\varepsilon^2)$ iterations.
\end{corollary}

\section{Generalization to Mirror Descent in the Bregman Setup}\label{sec:bregman proximal}
A \emph{distance-generating function} is $\phi:\+X \rightarrow \-R$ a continuous differentiable and strictly convex function. The \emph{Bregman divergence} associated with $\phi$ is defined as 
\begin{align*}
    D_\phi(x|y) = \phi(x) - \phi(y) - \InAngles{\nabla \phi(y), x - y}, \forall x,y \in \+X.
\end{align*}

\begin{definition}[Bregman Proximal Operator]
    A \emph{Bregman proximal operator} associated with $f \in \+F_{\lsc}(\+X)$ and $\phi$ is any mapping $\prox^\phi_f: \+X \rightarrow \+X$ such that for all $x \in \+X$,
    \begin{align*}
        \prox^\phi_f(x) \in \argmin_{y \in \+X} \{f(y) + D_\phi(y|x) \}.
    \end{align*}
\end{definition}

For simplicity, we assume $\phi$ is $1$-strongly convex with respect to some norm $\InNorms{\cdot}$\footnote{Here the norm $\InNorms{\cdot}$ is not necessarily $\ell_2$-norm and we denote its dual norm as $\InNorms{\cdot}_\star$.} in this section\footnote{For $\alpha$-strongly convex $\phi$, we can use $\frac{1}{\alpha}\phi$.}. Similar to the Euclidean case, the minimum in the definition of $\prox^\phi_f$ is attained when $\+X$ is compact, since $f$ is lower semi-continuous and $D_\phi(\cdot|x)$ is continuous. The minimizer may not be unique, and all of our results hold for every choice of the Bregman proximal operator. We define the \emph{Bregman--Moreau envelope} of $f$ as
\begin{align*}
    M^\phi_f(x) := \min_{y \in \+X}\{f(y) + D_\phi(y|x)\} = f(\prox^\phi_f(x)) + D_\phi(\prox^\phi_f(x)|x), \quad \forall x \in \+X.
\end{align*}

Examples of $1$-strongly convex distance-generating functions include
\begin{itemize}
    \item $\phi(x) = \sum_{i=1}^d x_i \log x_i$ defined on the simplex $\Delta^d$, which is $1$-strongly convex with respect to the $\ell_1$-norm. Then $D_\phi(x|y) = \sum_{i=1}^d x_i \log\frac{x_i}{y_i}$ is the KL divergence defined for all $x \in \Delta^d$ and $y$ in the relative interior of $\Delta^d$.\footnote{Since the negative entropy is differentiable only in the relative interior of $\Delta^d$, the Bregman divergence $D_\phi(\cdot|x)$, the Bregman proximal operator, and the Bregman--Moreau envelope are defined for $x$ in the relative interior of $\Delta^d$, which contains all iterates of \hyperref[MD]{Mirror Descent} initialized in the relative interior.}
    \item $\phi(x) = \frac{1}{2}\InNorms{x}^2$. Then $D_\phi(x|y) = \frac{1}{2}\InNorms{x-y}^2$ is the squared Euclidean norm.
\end{itemize}
The proximal operator in \Cref{dfn:proximal operator} is a special case of Bregman proximal operator with $\phi(x) = \frac{1}{2}\InNorms{x}^2$.

Similarly, we define the proximal regret associated with $\phi$ and $f$ as follows:
\begin{align}
    \reg^\phi_f(T) := \sum_{t=1}^T \ell^t(x^t) - \ell^t(\prox^\phi_f(x^t)). \tag{Bregman proximal regret}
\end{align}
We show that over a compact convex set $\+X$, \hyperref[MD]{Mirror Descent (MD)} associated with $\phi$ has Bregman proximal regret $\reg^\phi_f(T) = O(\sqrt{T})$ for all $f \in \+F_{\lsc}(\+X)$ simultaneously.

\begin{algorithm}[!ht]\label{MD}
    \KwIn{strategy space $\+X$, distance generating function $\phi$, non-increasing step sizes $\eta_t > 0$}
    \caption{Mirror Descent (MD)}
    Initialize $x^1 \in \+X$ arbitrarily\\
    \For{$t = 1,2, \ldots,$}{
    play $x^t$ and receive $\nabla \ell^t(x^t)$.\\
    update $x^{t+1} = \argmin_{x \in \+X}\{ \InAngles{ \eta_t \nabla \ell^t(x^t), x} +D_\phi(x|x^t) \}$.
    }
\end{algorithm}

\begin{theorem}\label{theorem:MD proximal regret}
    Let $\+X \subseteq \-R^d$ be a closed convex set and $\phi$ a $1$-strongly convex distance generating function w.r.t. norm $\InNorms{\cdot}$. Let $\{\ell^t: \+X \rightarrow \-R\}$ be a sequence of convex loss functions.
    Let $\{x^t\}$ be the sequence generated by  \hyperref[MD]{Mirror Descent} with $\phi$ and non-increasing step sizes $\{\eta_t\}$. Then for any $f \in \+F_{\lsc}(\+X)$ and any Bregman proximal operator $\prox^\phi_f$, denoting $p^t =\prox^\phi_f(x^t)$, we have
    \begin{align*}
        \sum_{t=1}^T \ell^t(x^t) - \ell^t(p^t) &\le \sum_{t=1}^T \InAngles{\nabla \ell^t(x^t), x^t - p^t}\\
        &\le \sum_{t=1}^T \frac{M^\phi_f(x^t) - M^\phi_f(x^{t+1})}{\eta_t} + \sum_{t=1}^T \frac{\eta_t}{2}\InNorms{\nabla \ell^t(x^t)}_\star^2.
    \end{align*}
    In particular, if $\+X$ is compact, then the first term on the right-hand side is at most $\frac{\max_{z \in \+X} D_\phi(z|x^1)}{\eta}$ when the step size is constant $\eta_t = \eta$, and at most $\frac{R^2}{\eta_T}$ in general, where $R^2 := \max_{t \in [T]}\max_{z \in \+X} D_\phi(z|x^t)$.
\end{theorem}
\begin{remark}
    For $\phi(x) = \frac{1}{2}\InNorms{x}^2$, \Cref{theorem:MD proximal regret} recovers \Cref{theorem:GD proximal regret}. For general $\phi$, the bound coincides with the classical external regret bound of \hyperref[MD]{MD}. For example, for the negative entropy $\phi(x) = \sum_{i=1}^d x_i \log x_i$ over the simplex $\+X = \Delta^d$ with the uniform initialization $x^1 = (\frac{1}{d}, \ldots, \frac{1}{d})$, we have $\max_{z \in \Delta^d} D_\phi(z|x^1) = \log d$. Thus, with constant step size $\eta$, the Bregman proximal regret is at most $\frac{\log d}{\eta} + \frac{\eta}{2}\sum_{t=1}^T\InNorms{\nabla \ell^t(x^t)}_\infty^2$ simultaneously for all $f \in \+F_{\lsc}(\Delta^d)$, which matches the classical regret bound of the Hedge algorithm.
\end{remark}
\begin{proof}[Proof of \Cref{theorem:MD proximal regret}]
    Denote $g^t = \nabla \ell^t(x^t)$. The first inequality follows from the convexity of $\ell^t$. By the standard analysis of \hyperref[MD]{Mirror Descent} (see e.g., \citep[Lemma 6.10]{orabona2019modern}), for any $u \in \+X$ and $t \in [T]$, we have
    \begin{align*}
        \eta_t \InAngles{g^t, x^t - u} \le D_\phi(u|x^t) - D_\phi(u|x^{t+1}) + \frac{\eta_t^2}{2}\InNorms{g^t}_\star^2.
    \end{align*}
    We choose $u = p^t$. Since $p^t$ is a minimizer in the definition of $M^\phi_f(x^t)$, we have $M^\phi_f(x^t) = f(p^t) + D_\phi(p^t|x^t)$; since $p^t \in \+X$ is a feasible point of the minimization problem defining $M^\phi_f(x^{t+1})$, we have $M^\phi_f(x^{t+1}) \le f(p^t) + D_\phi(p^t|x^{t+1})$. Therefore,
    \begin{align*}
        D_\phi(p^t|x^t) - D_\phi(p^t|x^{t+1}) \le M^\phi_f(x^t) - M^\phi_f(x^{t+1}).
    \end{align*}
    Combining the above two inequalities, dividing both sides by $\eta_t$, and summing over $t \in [T]$ proves the first part of the theorem.

    Now suppose $\+X$ is compact. Since $f$ is proper and lower semi-continuous, it attains its minimum over $\+X$ at some $z^\star \in \+X$. For any $x \in \+X$, since $D_\phi \ge 0$, we have $f(z^\star) \le M^\phi_f(x) \le f(z^\star) + D_\phi(z^\star|x)$. If the step size is constant, then the first term telescopes to $\frac{M^\phi_f(x^1) - M^\phi_f(x^{T+1})}{\eta} \le \frac{D_\phi(z^\star|x^1)}{\eta} \le \frac{\max_{z \in \+X} D_\phi(z|x^1)}{\eta}$. In general, let $e_t := M^\phi_f(x^t) - f(z^\star)$, so that $0 \le e_t \le R^2$ for all $t \in [T]$ and $e_{T+1} \ge 0$. Since $\{\eta_t\}$ is non-increasing, the same calculation as in the proof of \Cref{theorem:GD proximal regret} gives $\sum_{t=1}^T \frac{e_t - e_{t+1}}{\eta_t} \le \frac{R^2}{\eta_T}$. This completes the proof.
\end{proof}
\end{document}

%% file: macro.tex
\newtheorem{definition}{Definition}
\newtheorem{theorem}{Theorem}
\newtheorem*{theorem*}{Theorem}
\newtheorem{lemma}{Lemma}
\newtheorem{fact}{Fact}
\newtheorem{corollary}{Corollary}

\newtheorem{proposition}{Proposition}
\newtheorem{remark}{Remark}

\newcommand{\reg}{\mathrm{Reg}}

\usepackage{pifont}

\DeclareMathOperator*{\argmin}{argmin}

\DeclareMathOperator{\poly}{poly}

\def\+#1{\mathcal{#1}}
\def\-#1{\mathbb{#1}}

\newcommand{\notshow}[1]{{}}
\newcommand{\AutoAdjust}[3]{{\mathchoice{ \left #1 #2  \right #3}{#1 #2 #3}{#1 #2 #3}{#1 #2 #3}}}
\newcommand{\Xcomment}[1]{{}}

\newcommand{\InParentheses}[1]{\AutoAdjust{(}{#1}{)}}
\newcommand{\InBrackets}[1]{\AutoAdjust{[}{#1}{]}}
\newcommand{\InAngles}[1]{\AutoAdjust{\langle}{#1}{\rangle}}
\newcommand{\InNorms}[1]{\AutoAdjust{\|}{#1}{\|}}
\renewcommand{\part}[2]{\frac{\partial #1}{\partial #2}}

\newcommand{\R}{\mathbbm{R}}

\newcommand{\X}{\mathcal{X}}

\newcommand{\prox}{\mathrm{prox}}